\documentclass[a4paper,12pt]{article}

\usepackage[linkcolor=blue,citecolor=blue,colorlinks=true]{hyperref}

\usepackage{authblk}
\usepackage{amsmath}
\usepackage{amsfonts}
\usepackage{amssymb}
\usepackage{amsthm}
\usepackage[dvipsnames]{xcolor}
\usepackage{color}
\usepackage[font=small,labelfont=bf]{caption}

\newcommand{\oti}{}

\newcommand{\MRW}{M_{\mbox{\tiny RW}}}
\newcommand{\gRW}{g_{\mbox{\tiny RW}}}

\newcommand{\bm}[1]{\mbox{\boldmath $#1$}}

\def\defi{:=}

\def\onman#1{{\bm{#1}}}

\def\dim{n}
\def\dimu{\dim +1}

\def\signature{\varepsilon}
\def\curvN{\epsilon_0}

\def\UU{\mathcal{U}}
\def\intervalI{I}
\def\intervalIo{I_\oo}
\def\spaceN{\Sigma}

\def\SigmaUU{\Sigma}
\def\SigmaUUo{\SigmaUU}
\def\domainl{\mathcal{V}_\oo}
\newcommand{\oo}{{\mathfrak o}}

\def\GammaUU{\Upsilon}
\def\gS{\gamma}
\def\g{\gamma}

\def\GamVV{Q}
\def\expansion{\Theta}
\def\expanshom{\theta}

\def\ptos{=}
\def\Xtfol{X}
\def\Ytfol{Y}

\def\W{W_{\oo,X}}

\def\flow{U}

\def\newX{X'}

\def\X{{\mathfrak X}}

\def\ckvS{\zeta}

\def\FF{F}

\def\kill{k}

\def\UUumbilic{\varkappa}
\def\curvS{\kappa}
\def\curv{k}
\def\UURW{\Xi}

\def\Yy{Y}
\def\YyS{y}
\def\lS{x}

\def\xsq{|x|^2}
\def\ysq{|\YyS|^2}

\def\lsq{|\lambda|^2}
\def\zsq{|z|^2}

\def\ta{\alpha}
\def\tb{\beta}
\def\tc{\mu}
\def\td{\nu}

\def\Hess{\mathrm{Hess}}
\def\grad{{\rm{grad}}}

\def\gradS{\grad_\gS}
\def\nablag{\nabla^{g}}

\def\conffactor{\left(1+\frac{\onman\curv}{4}\lsq\right)}

\newtheorem{theorem}{Theorem}[section]
\newtheorem{lemma}[theorem]{Lemma}
\theoremstyle{definition}
\newtheorem{definition}[theorem]{Definition}

\newtheorem{remark}[theorem]{Remark}
\numberwithin{equation}{section}
\theoremstyle{plain}

\newtheorem{proposition}[theorem]{Proposition}

\newcounter{mnotecount}

\newcommand{\mnote}[1]
{\protect{\stepcounter{mnotecount}}$^{\mbox{\footnotesize $\bullet$\themnotecount}}$ 
\marginpar{
\raggedright\tiny\em
$\!\!\!\!\!\!\,\bullet$\themnotecount: #1} }

\begin{document}

\title{New characterization of Robertson-Walker geometries involving
a single timelike curve}

\author[1]{Marc Mars \thanks{marc@usal.es}}
\author[2]{Ra\"ul Vera \thanks{raul.vera@ehu.eus}}
\affil[1]{Instituto de F\'{\i}sica Fundamental y Matem\'aticas, IUFFyM,
Universidad de Salamanca, 
Plaza de la Merced s/n, 37008 Salamanca, Spain}
\affil[2]{Fisika Saila, Euskal Herriko Unibertsitatea UPV/EHU,\protect\\
  48080 Bilbao, Basque Country}

\maketitle

\begin{abstract}
  Our aim in this paper is two-fold. We establish a novel geometric characterization of the Roberson-Walker (RW) spacetime and, along the process, we find a canonical form of the RW metric associated to an arbitrary timelike curve and an arbitrary space frame.
  A known characterization establishes that a spacetime
  foliated by constant curvature leaves whose orthogonal flow
  (the cosmological flow) is geodesic, shear-free,
  and with constant expansion on each leaf, is RW.
  We generalize this characterization by relaxing the condition
  on the expansion. We show it suffices to demand that
  the spatial gradient and Laplacian of the cosmological expansion
  on a single arbitrary timelike curve vanish.
  In General Relativity these local conditions
  are equivalent to demanding that the energy flux measured
  by the cosmological flow, as well as its divergence, are zero
  on a single arbitrary timelike curve.
  The proof allows us to construct canonically adapted coordinates
  to the arbitrary curve, thus well-fitted to an observer
  with an arbitrary motion with respect to the cosmological flow.
\end{abstract}

\tableofcontents
\section{Introduction}

The basis for the standard model of the Universe is the Robertson-Walker (RW)
spacetime. The geometric nature of gravity makes the geometric characterization of spacetimes of great importance. Many efforts have been put in obtaining physically meaningful characterizations of the RW spacetime. Different characterizations probe different aspects the theory. Thus, finding new characterizations 
may be valuable to examine the foundations of the model.
We refer the reader to the extensive review in \cite{RelCosm}.

The characterizations of the RW geometry involve several kinds of
ingredients with different geometrical/physical/observational weight.
The first basic ingredient is the existence of isometries \cite{ONeill2,SolutionsBook}.
Another key ingredient is ``kinematical'',
in terms of the existence of a family of fundamental observers
defined as the (cosmological) flow $\flow$
moving with the average matter in the Universe,
to which certain requirements are imposed.
The dynamical ingredients 
are represented  by the Einstein
field equations once the energy-momentum tensor $T$ is constrained
by any means.
Other ingredients, that are ultimately linked
with the above, involve conditions on the Weyl tensor
or on the relation between the angular-diameter (or luminosity) 
distance and redshift \cite{Hasse-Perlick}.
The latter is an example where observational
constraints due to the isotropy of the measurements can be used to restrict the model \cite{Maartens2011}. Another example is the isotropy of the
Cosmic Microwave Background radiation that leads to the Ehlers-Geren-Sachs \cite{EGS} characterization of RW.

The different known characterizations resort on a combination of various of the above items
depending on which observational facts and principles, e.g. 
Cosmological or Copernican, are used to motivate them.
The paradigmatic characterization is due to Ellis \cite{Ellis1971}
and states that RW is characterized by the energy-momentum $T$ being isotropic (and thus of a perfect fluid type)
and the flow $\flow$ along the timelike eigendirection of $T$
being irrotational, shear-free and geodesic. Another characterization (see \cite{Krasinski} [Section 1.3.2]) that will play an important role later states that  a spacetime that
(a) is foliated by spaces of constant curvature and (b)
the orthogonal flow is geodesic, shear-free and the expansion is  constant
on each leaf, must be RW.

Let us note that in some cases not all the assumptions involved in the characterizations are spelled out in full detail or with sufficient accuracy.
Regarding the use of isometries,
we refer the reader to \cite{Avalos2022} for a  critical
analysis on the characterizations based mainly on the existence of different notions of isotropies.

It must be also stressed that,
although the aim of the mentioned results
(and the present work) is to characterize FLRW,
there exist alternatives to the
Cartan-Karlhede characterization scheme (see Chapter 9 in \cite{SolutionsBook})
or modifications thereof that provide, in particular,
characterizations for different families of FLRW geometries. That
can be used to consider the equivalence problem between different FLRW.
The IDEAL\footnote{Acronym for `Intrinsic, Deductive, Explicit and ALgorithmic' introduced in \cite{Ferrando2010}.} characterization of FLRW
\cite{Canepa2018} is based entirely on Rainich-type tensorial equations
(see \cite{Coll1989} in a more general setting) without relying on
the geometric structure of any (timelike) congruence.
On the other hand, the classification algorithm constructed in \cite{Wylleman2019} is based
on invariant quantities relative to a congruence of fundamental observers.

All characterizations based on isotropy properties around one single fundamental observer rely ultimately on the Copernican principle in order to extend these properties  to all fundamental observers, thus to the cosmological flow $\flow$. In this way, the properties are made global throughout the Universe and the conclusion that the geometry is RW  everywhere can be reached. In fact, all characterizations of RW spaces need a certain amount of global properties. In this paper we aim at finding a characterization result that splits the hypotheses in two sets, namely conditions that hold everywhere and conditions that hold on a single curve. As already said, global hypotheses are unavoidable, but in this paper we aim at weakening them. The assumption of global type that we make concern the geometry of the slices of constant cosmological time, i.e. the spaces orthogonal to $\flow$. Specifically we require that every slice at constant time $t$ has constant curvature $k(t)$, that its second fundamental form is pure trace and that the cosmological time has constant lapse. The last two conditions can be stated equivalently in terms of  $\flow$  by demanding that this field is shear-free and geodesic (compare with item (b) in the characterization above).
Let us stress that the curvature of the slices as a function of
  time $k(t)$ is left free, and can, in principle, change sign. This fact
  is clearly consistent locally, and there are even explicit global smooth constructions \cite{Sanchez2023}.

The  distinctive feature of our characterization  is the presence of a {\it local} condition based on a single timelike curve $\oo$. This curve is completely arbitrary so it need not correspond to a fundamental observer (although, of course, it may be). In fact, we can even dispense of the condition that the curve is timelike and consider any curve with the only restriction that it is
always transverse to the cosmological time. The condition that we impose is that the mean curvature of the constant time slices have vanishing gradient and vanishing Laplacian along~$\oo$. An equivalent way of stating this is that the expansion of the cosmological flow
$\flow$ has vanishing gradient and vanishing Laplacian along the chosen curve. Note that this condition is  much weaker than imposing that the expansion is homogeneous in space.
Interestingly, the local condition that we impose on the curve $\oo$ can also be stated as a restriction on the energy-momentum contents. Assuming the validity of the Einstein field equations, this condition can be rephrased as demanding that the energy flux measured by the fundamental observer {\it along the curve $\oo$} vanishes, together with its divergence. We find it remarkable that the global assumptions can be relaxed significantly and the extra parts be replaced by a very weak and physically reasonable condition on a single curve.

The precise statement of this characterization result appears in Theorem \ref{main1} below. We provide here a slightly more informal version.

\begin{theorem}
  \label{main_informal}
  Let $(M,g)$  be a spacetime of the form $M = I \times \Sigma$, where
  $I\subset \mathbb{R}$ is an open interval. Assume that each $\SigmaUU_t := \{ t\} \times \Sigma$ is spacelike and let $g_t$ be the induced metric and
  $\flow$  its future directed unit normal. Assume that $g_t$ has constant curvature $k(t)$ and that $\flow$ is shear-free and geodesic. Then $(M,g)$ is a Robertson-Walker space if and only if there is a timelike curve $\oo(t)$ along which the expansion $\expansion$ of $\flow$ satisfies
  \begin{align}
    \grad_{g_t} \expansion |_{\oo(t)}=0, \qquad \Delta_{g_t} \expansion |_{\oo(t)} =0.
  \end{align}
  \end{theorem}

  The main idea behind the proof of this theorem is to show that the local condition  along the curve together with the global assumptions are sufficient to prove that that expansion $\expansion$ of $\flow$ must be constant on each slice $\SigmaUU_t$. We can then rely on the known characterization mentioned above to reach the conclusion.

  The method of proof is based on the construction of stereographic coordinates at each
  $\SigmaUU_t$ centered at the point $\oo(t)$ and associated uniquely to an arbitrary choice of orthonormal frame $\{X_A(t)\}$ at $\oo(t)$. These coordinates not only play a crucial role in the proof, but they also have an added bonus. Once the characterization result is obtained, we have as a by-product the form of the  RW metric written in terms of coordinates canonically constructed from an arbitrary timelike curve $\oo(t)$ (actually any tranverse curve) and an arbitrary choice of orthonormal space frame $\{ X_A(t)\}$ along $\oo$. The form of the metric involves, besides the
  scale factor $a(t)$ and the discrete curvature parameter $\curvN \in \{ -1,0,1\}$, three\footnote{Actually we work in arbitrary dimension, but for simplicity we restrict to the four-dimensional case in this description.}
  functions $\FF_{0A} (t)$ and three functions $\FF_{AB} (t ) = - \FF_{BA}(t)$ associated to the freedom in the curve $\oo(t)$ and the frame
  $X_{A}(t)$. In fact, the functions $\FF_{0A}$ describe the
  velocity of the observer along the curve with respect to the
  cosmological flow $\flow$, and we show that by an appropriate choice of frame we can
  always set $\FF_{AB}(t)=0$ if so desired. The precise form of the metric in these coordinates is the content of Theorem \ref{main2}. The form of the metric is of course more complicated than any of the standard ones, but the fact that it is canonically adapted to a single arbitrary observer  makes it interesting and potentially useful in problems in Cosmology where a privileged observer not at rest with the cosmological flow is involved,
as it happens with our direct observations \cite{Lineweaver1996a}
(e.g. any Earth-based or satellite telescope).
To be precise, the dipole contribution to the CMB due to our peculiar velocity
with respect to the cosmological flow
can thus be incorporated at the background level of
a perturbative approach by setting the peculiar
  velocity as the three-velocity $v^A=-a^{-1}(t)F_{0B}(t)\delta^{AB}\partial_{z^B}$
  in the coordinates of  Theorem \ref{main2} (see Remark \ref{velocity}).
That could help on the disentangling the
degeneracy of the Doppler effect, due to the peculiar velocity
and a dipolar part of the perturbations (see e.g. \cite{Roldan2016}).

  The plan of the paper is as follows. In Section \ref{Definitions} we give the precise definition of RW space that we shall use. The definition is standard except that we allow for both Lorentzian and Riemannian signature, since this generality entails essentially no extra effort. We then quote a known characterization result  of RW geometries that will play a key role in our argument. In Section \ref{sec:geo_H} we describe the geometry of spaces that admit a foliation of constant curvature and umbilic leaves. The geometric framework in this section is more general than in the rest of the paper. Although we could have simplified this part, there are several applications that we have in mind where the more general framework is needed. The key result of this section
  is Proposition \ref{Umbilicfoliation} where canonical coordinates adapted to a transverse curve $\oo(t)$ and an orthonormal frame are obtained. On each leaf of the foliation these coordinates are stereographic, so in order to make the paper self-contained we provide in Appendix \ref{sec:constant_curv_spaces} a very simple and direct construction of stereographic coordinates on any Riemannian space of constant curvature.
  Section  \ref{NewCharac} contains our main results.  Theorem \ref{main1} is the characterization result of RW involving a transverse curve. Theorem \ref{main2} provides the RW metric in canonical coordinates adapted to
  $\{ \oo(t)\}$ and $\{ X_A(t)\}$. Subsection  \ref{Description} is devoted to describing several geometric properties of these coordinates, in particular the relation of the functions $\FF_{AB}(t)$ to rotations of the frame $\{ X_A(t)\}$. In Subsection \ref{Sub:Killings}
  the explicit form the Killing vectors of RW in the newly constructed coordinates is presented.

\subsection{Basic notation}

The set of vector fields on a manifold $\UU$ is denoted by
$\X(\UU)$. We use $\pounds_X$ for the Lie derivative along a vector field $X \in \X(\UU)$.
Given a function $f \in C^{\infty}(M,\mathbb{R})$ in a semi-Riemannian space $(M,g)$ we use
$\grad_g f$ to denote its gradient vector field and
$\Delta_g f$ for the Laplacian. We use both abstract-index or index-free notation depending on our convenience. Capital Latin indices $A,B,C$ take values $1, \cdots, \dim$ and Greek indices  in $0,\cdots, \dim +1$. The symbol $\delta_{AB}$ stands for the Kronecker delta.

\section{Definitions and a previous characterization result}
\label{Definitions}
In this paper we adopt the following definition of an $\dim+1$ dimensional RW geometry:
 \begin{definition}[RW space]
\label{def:RW}
Let $\signature = \pm 1$, $\intervalI \subset \mathbb{R}$ an open interval and $(\spaceN,g_{\curvN})$ a (positive definite)
Riemannian manifold of dimension $\dim \geq 1$ and constant curvature $\curvN \in \{ -1,0,1 \}$\footnote{When $n=1$ there is no curvature, and we may choose the value  of $\curvN \in \{-1,0,1\}$ arbitrarily.}.
A \textbf{RW space}, denoted by $\signature I \times_a \spaceN$,  is the manifold $M = I \times \spaceN$
endowed with the warped product metric
$g = - \signature d\tau^2 + (a \circ \tau)^2 g_{\curvN}$, where $\tau \in C^{\infty}(M, I)$
is the projection to the first factor and 
$a \in C^{\infty}(I,\mathbb{R}^+)$ is the warping function.
\end{definition}
This definition is standard except for the fact that we are allowing any sign in $\signature$, which means that $g$ can be either of Lorentzian ($\signature=1$) or Riemannian ($\signature =-1$) signature. Our primary interest is in the Lorentzian case, but with essentially the same effort we can deal with both cases.
Note that we are not imposing any global assumption on the base space
  $(\spaceN, g_{\curvN})$, such as
  completeness or simply connectednes.

This definition is sufficiently general for our characterization purposes, and it is also well-adapted to a known characterization result that will play a relevant role later. Characterization results can be either local or global. One speaks of a local characterization when the assumptions made on the space $(M,g)$ under consideration  are sufficient to show that there exists a
  RW space $(\MRW,\gRW)$ such that,  at  every point $p \in M$ there is an open
  neighbourhood of $W_p$ of $p$ and a map $\Phi_p : (W_p,g|_{W_p}) \to  (\MRW,\gRW)$ which is an isometry from $(W_p,g|_{W_p})$ onto its image. One simply says that
  $(M,g)$  is {\it locally} a RW space. The characterization is global if there exists an isometry $\Phi: (M,g) \to (\MRW,\gRW)$. Obviously, global characterizations are stronger, so they also need stronger assumptions.

There are several local and global characterization results of RW spaces.
We refer to \cite{Avalos2022} for a detailed recent account that discusses many of them (and corrects some misleading or incorrect statements in the literature). The one that will be relevant for this paper was stated in rigorous form in \cite{Miguel_RW} in the more general context of generalized RW spaces. This result had a local and a global version. The global version was then improved in \cite{Caballero2011} building on previous results
in \cite{GutierrezOlea}. The statements of those papers needed here are summarized in
the following theorem (see also Proposition 1 in \cite{Carot1993}, and
\cite{ColeyMcManus}, \cite{Krasinski}).
\begin{theorem}[Local result from Theorem 2.1 in \cite{Miguel_RW} and its remark; global result from Theorem 3.1 in \cite{Caballero2011} ]
  \label{Miguel}
 Let $(M,g)$ be a connected Lorentzian ($\signature=1$)
  or Riemannian ($\signature=-1$) manifold
  endowed with a function $\tau\in C^{\infty}(M,\mathbb{R})$
with non-zero gradient everywhere and let $\flow$ be a vector field
 orthogonal
 to the hypersurfaces of constant $\tau$ normalized to be unit, i.e.
 $g(\flow,\flow)=-\signature$. Assume that $\flow$ is 
     
\begin{enumerate}
 \item[(i)] geodesic,  i.e. $\nabla_{\flow} \flow=0$,
\item[(ii)] shear-free, i.e.
  $\pounds_{\flow} g(X,Y)=\UUumbilic g(X,Y)$, for all $X,Y$ orthogonal to
  $\flow$, for some function $\UUumbilic\in C^\infty(M,\mathbb{R})$,
\item[(iii)]   
  with $\grad_g\,\UUumbilic$ being  pointwise parallel to $\flow$.
\end{enumerate}
Suppose that the hypersurfaces defined by the level sets of $\tau$
  are of constant curvature.
  Then $(M,g)$ is locally a RW space.

  Assume further that, for some $t_0 \in \mathbb{R}$ there exists
    an interval $I \subset \mathbb{R}$ such that the map $\Phi: I \times
    \Sigma_{t_0}
    \to M$ defined as the flow of the vector field
    $\flow$ is well-defined and onto. Then
    $(M,g)$ is globally
  a RW space.\footnote{\label{fn1}This global result is stated in \cite{Caballero2011} only in the
      Lorentzian case, but its generalisation to the Riemannian case is immediate.}
\end{theorem}

\begin{remark}
    Note that, because of (i) and (ii), the covector $\bm{\flow} := g (\flow,\cdot)$ satisfies $d \bm{\flow}  =0$,
    and therefore there is a function $\tilde\tau$ so that $\flow=\grad_g\tilde\tau$, locally.
    \end{remark}

  \begin{remark}
    Theorem 3.1 in \cite{Caballero2011} is stated differently. Hypotheses up to (iii) are replaced by the condition that $(M,g)$ admits a gradient
    timelike conformal Killing vector field, namely a vector field $\flow'$ which is the gradient of a function $\tau' \in C^{\infty}(M,\mathbb{R}) $ (in general, different from $\tau$ above) and satisfies
    \begin{align*}
    \signature g (\flow',\flow') <0, \qquad \pounds_{\flow'} g = 2 \rho g,
    \qquad \rho \in C^{\infty}(M,\mathbb{R}).
  \end{align*}
  It is easy to show that the two sets of conditions are equivalent.
    \end{remark}

\section{Umbilic foliations with constant-curvature leaves}
\label{sec:geo_H}

Although our  characterization theorem involves Lorentzian or Riemannian spaces, it is of interest to derive some of the intermediate results in a more general setup. 
As already advanced in the Introduction, this is convenient for future applications to other problems.  In this section we describe this more general setup.

Throughout this section $(\UU, \intervalI, \tau)$
denotes an $(\dimu)$-dimensional manifold, $\intervalI \subset \mathbb{R}$ an open interval and
$\tau$ a smooth function 
$\tau:\UU\to\intervalI$
such that $d\tau$ does not vanish anywhere. For each $t \in \intervalI$ the level surface of $\tau$, $\SigmaUU_t\defi\tau^{-1}(t)$ is a hypersurface of $\UU$
and the set of all $\{ \SigmaUU_t\}$ foliate $\UU$ (see e.g. \cite{ONeill2}).
The submanifolds $\SigmaUU_t$ are  always embedded, but in general they not need to be
diffeomorphic to each other, or connected. However, in sufficiently small open neighbourhoods of a point $p \in \UU$ this is always the case.
We  call $(\UU, \intervalI, \tau)$ a {\bf foliated space}.

The following notation will be used in foliated spaces.
The symbol $t$ always denotes a value in $\intervalI$. For any $f \in C^{\infty}(\UU,\mathbb{R})$, we use $f_t$ for the restriction of $f$ to $\SigmaUU_t$.
More precisely $f_t:=i^\star_t(f)$ where  $i_t:\SigmaUU_t\hookrightarrow\UU$ is the inclusion map.
Any  function $F \in C^{\infty}(\intervalI,\mathbb{R})$ can be transferred to a function on $\UU$ which we denote with the same symbol but in boldface font. Specifically, $\onman{F}\defi F\circ \tau$.
A vector field $\Xtfol \in \X(\UU)$ is said to be tangent to the foliation iff
$\Xtfol(\tau)=0$. For such vector fields, there always exists a  unique vector
field $X_t \in \X(\SigmaUU_t)$ such that $d i_t (X_t ) = X |_{\SigmaUU_t}$.

We shall consider foliated spaces $(\UU, \intervalI,\tau)$ that carry an additional geometric structure capable of inducing 
a positive definite Riemannian metric on each $\SigmaUU_t$. We therefore
assume that $\UU$ is endowed with  a symmetric
two-covariant tensor field $\GammaUU$ such that
$\GammaUU{}_t:=i_t^\star(\GammaUU)$ is a positive definite metric.
For this paper, it would suffice to assume that $\GammaUU$ is either a Riemannian metric, or a Lorentzian metric for which $\SigmaUU_t$ are spacelike hypersurfaces, but we keep this more general setup for future purposes.
We call the collection $(\UU,\intervalI,\tau,\GammaUU)$ a {\bf metric foliated space}.

In a metric foliated space there is a unique vector field $V$ defined by the properties (i) $V(\tau)=1$ and (ii) $\GammaUU(V,\Xtfol)=0$ for all vector fields
$\Xtfol \in \X(\UU)$ tangent to the foliation. Indeed, at any point $p \in \UU$ the tensor
$\GammaUU|_p$ is positive definite when restricted to the $n$-dimensional
subspace of vectors tangent to $\SigmaUU_t$. Thus, the signature of $\GammaUU|_p$ is necessarily $\{ \epsilon_p ,+, \cdots, +\}$ where
$\epsilon_p \in \{ -1,0,1\}$ (note that  $\epsilon_p$ may change from point to point). By basic linear  algebra, the space
\begin{equation*}
(T_p \SigmaUU_t)^{\perp} :=
  \{ X \in T_p \UU;  \quad \GammaUU|_p(X, \Ytfol)= 0 \quad \forall \Ytfol \in T_p \SigmaUU_t\}
\end{equation*}
is one-dimensional. All non-zero vectors  $Z \in (T_p \SigmaUU_t)^{\perp}$
are transverse to
$\SigmaUU_t$, i.e. satisfy $Z(\tau) \neq 0$ (otherwise we would have a non-zero vector tangent to $\SigmaUU_t$ and orthogonal to all other tangent vectors, which is impossible given that 
$\GammaUU_{t}$ is positive definite). It is therefore clear that there is exactly one vector $V|_p \in (T_p \SigmaUU_t)^{\perp}$ satisfying $V|_p(\tau)=1$. It is immediate to check that $V$ depends smoothly on $p$ and hence defines a vector field.

We shall call  $V$ the {\it normal vector field of the foliation}.
Observe that $V$ depends on the choice of $\tau$ and not just on the leaves
$\{ \SigmaUU_t\}$ of the foliation. Another function $\tau' : \UU \rightarrow
\intervalI'$ will define exactly the same leaves as $\tau$ if and only if 
$\tau' =  F \circ \tau$ where $F: I \rightarrow I'$ is a diffeomorphism mapping the interval $I$ to the interval $I'$.
The normal field
$V'$ of $\tau'$ is related to the normal field $V$ of $\tau$ by
$V' =  ( \onman{\frac{d F}{dt}})^{-1}  V$. This property will be used below.
To avoid any misunderstanding, note that we are making no assumption on the norm of $V$. Even when $\GammaUU$ is a metric, $V$ will in general {\it not} be
a unit normal to the hypersurfaces.

The two global assumptions we shall make in order
to characterize a RW space are encoded in the following definition:
\begin{definition}
  \label{Def:Umbilic}
  A metric foliated space $(\UU,\tau,\intervalI,\GammaUU)$
  is called an {\bf umbilic foliation with constant curvature leaves} iff:
  \begin{enumerate}
\item There exist a smooth function $\UUumbilic:
  \UU \rightarrow \mathbb{R}$ such that 
  \begin{equation}
  \label{eq:umbilic_fol}  
  (\pounds_{V}\GammaUU)(\Xtfol,\Ytfol)=\UUumbilic\,\GammaUU(\Xtfol,\Ytfol)
\end{equation}
for all vectors $\Xtfol$, $\Ytfol$ tangent to the foliation.
\item The induced metric $\GammaUU_t$ on each $\SigmaUU_t$ is a metric of constant curvature $\curv(t)$.
\end{enumerate}
\end{definition}
It will be convenient to distinguish when conditions 1 or 2  are being used. When only 1 holds we speak of an {\bf umbilic foliation} and
when only 2 holds of a {\bf foliation with constant curvature leaves}.

As described in the Introduction, the local assumption that we make to characterize a RW geometry involves a single timelike curve. In  the more general setup of this section we consider a smooth curve $\oo$ transverse to the foliation, i.e.  such that its tangent vector $\mathfrak{t}_\oo$ satisfies $\mathfrak{t}_{\oo} (\tau) \neq 0$
everywhere. Letting $\intervalIo \subset \intervalI$ be the image of
  $\tau \circ \oo$ (i.e. the set of values that $\tau$ takes along the curve), we may parametrize $\oo$ by $t \in \intervalIo$.
In other words, we consider the curve as described by a map
$\oo : \intervalIo\subset \intervalI \to\UU$ such that 
$\tau \circ \oo$ is the identity map of $\intervalIo$, or equivalently
$\oo(t) \in \SigmaUU_t$ for all $t \in \intervalIo$. By definition we say that the curve $\oo$ is {\it parametrized by $t$} whenever this happens.

The curve
$\oo$ describes the path followed by our privileged observer. We also need this observer to be endowed with a frame. To that aim, we select $\dim$ vector fields $X_A(t)$, $A,B = 1,\cdots, \dim$,
along the curve (i.e. such that $X_A (t) \in T_{\oo(t)} \UU$) smoothly depending\footnote{To make this smooth dependence precise, simply note that
  the curve $\oo$ carries an associated vector bundle where the fiber at each point
  $\oo(t)$ is the tangent space $T_{\oo(t)} \UU$. The map $X_A (t)$ is constructed
  simply by taking a section of this vector bundle and coordinating the base space with $t$.}
on $t$ with the properties that (i) $X_A(t)$ is 
tangent to
$\SigmaUU_t$ at $\oo(t)$ and (ii) $\{ X_A(t)\}$ is an orthonormal basis of $T_{\oo(t)} \SigmaUU_t$, namely $\GammaUU_t (X_A, X_B) = \delta_{AB}$.

From the curve $\oo$ and the orthonormal frame $\{ X_A(t)\}$ we can construct
a canonical coordinate system in a neighbourhood of $\oo$ as follows. For
each $t \in \intervalIo$ 
we have a point $\oo(t) \in \SigmaUU_t$ and
an orthonormal frame $\{ X_A(t)\}$ of $T_{\oo(t)} \SigmaUU_t$. The space
$(\SigmaUU_t,\GammaUU_t)$
is a space of constant curvature $\curv(t)$. In appendix \ref{sec:constant_curv_spaces} we give a self-contained and  direct
description of how to construct stereographic coordinates of
$(\SigmaUU_t,\GammaUU_t)$ centered at $\oo(t)$ and with frame
$\{ X_A(t)\}$. The procedure involves a so-called {\it Hessian basis}
(see Definition \ref{def:hess_basis}), which is a set of $\dim+2$
real valued functions $\Yy_t^{\ta}=\{\Yy_t^0,\Yy_t^A,\Yy_t^{\dimu}\}$ that solve the  PDE problem
\begin{align}
  \Hess_t\Yy_t^{\ta}= \left( -\curv(t) \Yy_t^{\ta} + \delta^{\ta}_{\dim+1}\right) \GammaUU_t \label{eq:Hesst} \\
  \Yy_t^{\ta} |_{\oo(t)} = \delta^{\ta}_0, \qquad 
  X_A (\Yy_t^{\ta}) |_{\oo(t)} = \delta^{\ta}_A,  \label{boundaryt} 
\end{align}
where  $\Hess_t$ is the Hessian with respect to $\GammaUU_t$.
Given that
$(\SigmaUU_t, \GammaUU_t)$ is of constant curvature $\curv(t)$,
Remark \ref{constant} in Appendix \ref{sec:constant_curv_spaces}
implies the existence of an open neighbourhood $ \SigmaUU'_t \subset \SigmaUU_t$ of $\oo(t)$ where the PDE above admits a unique solution.
Define $\UU' := \cup_{t \in \intervalIo} \SigmaUU'_t$. This is an
open neighbourhood of the curve $\oo$. On this set we can define
scalar functions $\Yy^{\ta}=\{\Yy^0,\Yy^A,\Yy^{\dimu}\}$ simply by
\begin{equation*}
  \Yy^{\ta} (p) = \Yy^{\ta}_{\tau(p)} (p), \qquad \forall p \in \UU',
\end{equation*}
i.e. we simply stack the functions $\{ \Yy_t^{\ta}\}$ together. Since the
boundary conditions are independent of $t$ (in particular smooth in $t$)
and the curve $\oo$ is also smooth, it follows that the functions
$\{ \Yy^{\ta}\}$ are smooth on $\UU'$.

For each value $t$ $\in \intervalIo$,  we now apply Lemma \ref{properties}
(together with  Remark \ref{constant}) to
$(S,\gS)=(\SigmaUU'_t,\GammaUU_t)$, with
curvature $\curvS=\curv(t)$,  point $o= \oo(t)$, basis $\{X_A = X_A(t)\}$ and functions $\YyS^{\ta} = \Yy_t^{\ta}$.
As a consequence,  the functions $\{\lambda^A\}$ defined  by
\begin{equation*}
  \lambda^A := \frac{\Yy^A}{1- \frac{1}{2} \onman{\curv} \Yy^{\dim+1}} 
\end{equation*}
are well-defined on some open domain $\domainl \subset \UU'$ containing $\oo$. We restrict ourselves to $\domainl$ from now on
  (and still use
$\SigmaUU_t$ to denote $\SigmaUU_t \cap \domainl$).

In terms of $\lambda^A$ the functions $\{\Yy^{\ta}\}$ take the form, c.f. \eqref{Yinx},
\begin{equation}
  \Yy^0=\frac{4-\onman{\curv}\lsq}{4+\onman{\curv} \lsq},\quad
  \Yy^A=\left(1+\frac{\onman{\curv}}{4}\lsq\right)^{-1}\lambda^A,\quad
  \Yy^{\dimu}=\left(1+\frac{\onman{\curv}}{4}\lsq\right)^{-1}\frac{1}{2}\lsq,
  \label{eq:Ys}
\end{equation}
where $\lsq:=\delta_{AB}\lambda^A\lambda^B$ since in the present case
$h_{AB}= \delta_{AB}$ given that the basis $\{X_A\}$ has been taken to be orthonormal. The restriction $\lambda^A_t$
of  $\lambda^A$ to each $\SigmaUUo_t$
are stereographic coordinates
of $\SigmaUU_t$ with center $\oo(t)$ and
frame $\{ X_A(t)\}$, so the metric takes the form (cf. \eqref{confflat})
\begin{equation}
  \GammaUU_t=\left(1+\frac{\curv (t)}{4}\lsq_t\right)^{-2}\delta_{AB}d\lambda_t^A d\lambda_t^B.
  \label{eq:gammat}
\end{equation}

The functions $\{ \Yy^{\ta}\}$ are smooth on $\domainl$, so the same holds for
$\{ \lambda^A\}$. Moreover, the set $\{\tau,\lambda^A\}$ defines a coordinate chart on $\domainl$. 
Since we know the form of the restriction of $\GammaUU$ to the
hypersurface $\{\tau = t\} \cap \domainl$, namely \eqref{eq:gammat}, the full tensor
$\GammaUU$ in coordinates $\{ \tau, \lambda^A\}$ must take the form
\begin{equation}
  \GammaUU= \GamVV d\tau^2+\conffactor^{-2}
  \delta_{AB}(d\lambda^A + f^A d\tau)(d \lambda^B+ f^Bd\tau)
  \label{eq:gamma_pre_pre}
\end{equation}
for some functions $\GamVV,f^A\in C^{\infty}(\domainl,\mathbb{R})$. We can also express the normal vector field of the foliation $V$ in this coordinate chart. The two conditions $V(\tau)=1$ and $\GammaUU(V,\partial_A) =0$ are equivalent to
\begin{equation*}
  V = \partial_{\tau} - f^A \partial_A.
\end{equation*}
Hence $f^A = - V (\lambda^A)$ and we can rewrite  \eqref{eq:gamma_pre_pre} as
\begin{equation}
  \GammaUU=\GamVV d\tau^2+\conffactor^{-2}
  \delta_{AB}(d\lambda^A -V(\lambda^A)d\tau)(d \lambda^B-V(\lambda^B)d\tau).
  \label{eq:gamma_pre}
\end{equation}
This expression gives us a handle on how to determine the form of the tensor
$\GammaUU$ in the context of Definition \ref{Def:Umbilic}.

So far we have only assumed that the foliation has constant curvature leaves.
We now impose the condition that the foliation is umbilic.
The stereographic coordinates have allowed us to construct
a vector field $\partial_{\tau}$ on $\domainl$. This vector field is canonical in the following sense. In geometric terms $\partial_{\tau}$ is the field of tangents to the curves
of constant $\lambda^A$ parametrized by $\tau$. Since the stereographic coordinates are canonically constructed from the curve $\oo$ and the basis $\{ X_A(t)\}$, the vector field depends {\it only} on $\oo$ and
$\{ X_A(t)\}$. It can therefore be denoted  by $\W$. We prefer this name over
$\partial_{\tau}$, as this emphasizes its geometric meaning. Obviously, in the
coordinates $\{ \tau, \lambda^A\}$ it holds $\W = \partial_{\tau}$.
The following proposition is key in the determination of $\GammaUU$.

\begin{proposition}
  \label{Vtangent}
  Let $(\UU,\intervalI,\tau,\GammaUU)$ be an umbilic foliation with constant curvature leaves as defined in Definition
  \ref{Def:Umbilic}. Let $\oo$ be a curve parametrized by $t$ and  $\{ X_A(t)\}$ an orthonormal basis along $\oo$. Consider the domain $\domainl$ as constructed above.
  
  Then the vector field $V - \W$ defined on $\domainl$
  is tangent to the foliation and
  the restriction to each leaf $\SigmaUUo_t$
  is a conformal Killing vector of $\GammaUU_t$
  with conformal factor $ \UUumbilic_t + \Yy_t^{\dimu}{\dot\curv(t)} $
    (c.f. \eqref{eq:conf_equation}).
\end{proposition}

\begin{proof}
  The fact that $V - \W$ is tangent to the foliation is immediate from $(V- \W)(\tau)=0$.  Now we compute
  $\pounds_{\W} \GammaUU$ acting on tangent vectors. In the coordinates $\{ \tau, \lambda^A\}$ we have that $\W =
  \partial_{\tau}$ and that $\partial_A$ is a basis of tangent vectors. So, it suffices to compute
  \begin{equation*}
    (\pounds_{\partial_\tau}\GammaUU)(\partial_A,\partial_B)
    =\partial_\tau(\GammaUU(\partial_A,\partial_B))
    =-\frac{2\lsq}{4+\onman{\curv}\lsq}\onman{\dot\curv}
    \GammaUU(\partial_A,\partial_B),
  \end{equation*}
  where $[\partial_{\tau},
  \partial_{\lambda^A}]=0$ was used in the first equality and
  \eqref{eq:gamma_pre} in the second. In more geometric terms, this equality can be written as
  \begin{equation*}
    i^\star_t(\pounds_{\W}\GammaUU)
    =-\frac{2 \lsq_t}{4+\curv(t)\lsq_t}{\dot\curv(t)}
    \GammaUU_t=-\Yy_t^{\dimu}{\dot\curv(t)}\GammaUU_t.
  \end{equation*}
  Concerning $V$, the umbilicity condition \eqref{eq:umbilic_fol} can be written
  as $i_t^{\star} (\pounds_V \GammaUU) = \UUumbilic_t
  \GammaUU_t$.      Hence
  \begin{equation*}
    i_t^{\star}
    \left ( \pounds_{V - \W} \GammaUU \right )
    = \left ( \UUumbilic_t + \Yy_t^{\dimu}{\dot\curv(t)} \right ) \GammaUU_t.
  \end{equation*}
  Now, for any  vector field $Z$ tangent to the foliation, i.e. of the form $Z = di_t (Z_t)$, $Z_t \in \X(\SigmaUU_t)$, and
  any covariant tensor field $T$ on $\UU$ the following
  general identity holds 
  \begin{equation*}
    i_t^{\star} \left ( \pounds_Z T  \right ) =
    \pounds_{Z_t} \left ( i_t^{\star}(T) \right ).
  \end{equation*}
  Applying this to $Z= V -\W$ and $T= \GammaUU$,
  and defining $V^\ptos_t \in \X (\SigmaUU'_t)$
  by $di_t (V^{\ptos}_t ) = V - \W$ we get
  \begin{equation}
    \pounds_{V^{\ptos}_t} \GammaUU_t =
    \left ( \UUumbilic_t + \Yy_t^{\dimu}{\dot\curv(t)} \right ) \GammaUU_t,
    \label{eq:conf_equation}
  \end{equation}
  which proves that $V^{\ptos}_t$ is a conformal Killing vector of $\GammaUU_t$
  with conformal factor $ \UUumbilic_t + \Yy_t^{\dimu}{\dot\curv(t)}$, as claimed.
\end{proof} 

In Appendix \ref{sec:constant_curv_spaces} we write down a basis of the  conformal Killing algebra of any  space of constant curvature  in terms of a Hessian basis (related results can be found in \cite{MarsPeonNieto}) and we find the explicit form of the basis vectors in stereographic coordinates centered at a point $o$ with basis $\{X_A\}$. This requires that the dimension of the space is at least three. So we assume $\dim \geq 3$ from now on.

Translating the results in
\ref{res:CKVs} into the present setting we get that the vector fields in $\SigmaUU'_t$ defined by
\begin{equation*}
  \eta_t^{\ta\tb} := \Yy_t^{\ta} \grad_{\GammaUU_t} \Yy_t^{\tb}
  - \Yy_t^{\tb} \grad_{\GammaUU_t} \Yy_t^{\ta}
\end{equation*}
are such  that  $\{\eta_t^{\ta\tb}, \ta < \tb\}$ is a basis of the conformal
Killing algebra of $(\SigmaUUo_t, \GammaUU_t)$.
Their explicit form in the coordinates   $\{\lambda^A_t\}$ is (cf. \eqref{ckvl_S}-\eqref{killsl_S})
\begin{align*}
  &\eta_t^{0\, \dimu}= - \eta_t^{\dimu\, 0} = \lambda_t^A\partial_A,\qquad
    \eta_t^{A\, \dimu} = - \eta_t^{\dimu\, A} 
    =\left(\lambda_t^A\lambda_t^B-\frac{1}{2}\lsq_t \delta^{AB} \right)\partial_B, \\
  &\eta_t^{AB}=(\lambda_t^A\delta^{BC}-\lambda_t^B\delta^{AC})\partial_C,\quad
    \eta_t^{0A}= - \eta_t^{A0} =
    \delta^{AB}\partial_B+\frac{\curv(t)}{2} \eta_t^{A\, \dimu}. 
\end{align*}
Moreover, these fields satisfy
\begin{equation}\label{eq:zeta_eqs}
 \pounds_{\eta_t^{0\,\dimu}}\GammaUU_t =2\Yy_t^{0}\GammaUU_t,\qquad
 \pounds_{\eta_t^{A\,\dimu}}\GammaUU_t=2\Yy_t^{A}\GammaUU_t,\qquad
 \pounds_{\eta_t^{0A}} \GammaUU_t =
 \pounds_{\eta_t^{AB}} \GammaUU_t = 0,
\end{equation}
so in particular $\{\eta_t^{0A}, \eta_t^{AB}\}$ with  $A <B$ is a basis of the Killing algebra of $(\SigmaUUo_t, \GammaUU_t)$. Let us define the vector fields $\eta^{\ta\tb}$
on $\domainl$ by $di_t ( \eta_t^{\ta\tb})$. In the coordinates
$\{ \tau, \lambda^A\}$ we clearly have
\begin{align}
  &\eta^{0\, \dimu}= - \eta^{\dimu\, 0} = \lambda^A\partial_A,\qquad
    \eta^{A\,\dimu} = - \eta^{\dimu\, A} 
    =\left(\lambda^A\lambda^B-\frac{1}{2}\lsq \delta^{AB} \right)\partial_B,\label{ckvl}\\
  &\eta^{AB}=(\lambda^A\delta^{BC}-\lambda^B\delta^{AC})\partial_C,\quad
    \eta^{0A}= - \eta^{A0} =
    \delta^{AB}\partial_B+\frac{\onman{\curv}}{2} \eta^{A\,\dimu}.\label{killslt}
\end{align}
Proposition \ref{Vtangent} implies the existence of $(n+1)(n+2)/2$
functions $C_0, C_A, \FF_{0A}, \FF_{AB} = - \FF_{BA} \in  C^{\infty} (\intervalIo,\mathbb{R})$   such that
\begin{align}
  V_t^\ptos& =
             C_0(t)\eta^{0\,\dimu}_t + C_A(t) \eta^{A\,\dimu}_t
             +\FF_{0A}(t)\eta_t^{0A}+ \frac{1}{2} \FF_{AB}(t)\eta_t^{AB}
             \qquad \Longleftrightarrow \nonumber \\
  V - \W &=
           \onman{C}_0 \eta^{0\,\dimu} + \onman{C}_A \eta^{A\,\dimu}
           + \onman{\FF}_{0A}\eta^{0A}+ \frac{1}{2} \onman{\FF}_{AB}\eta^{AB}.
           \label{eq:V=}
\end{align}
The conformal Killing vector $V^{\ptos}_t$
must satisfy (\ref{eq:conf_equation}), which after using 
(\ref{eq:zeta_eqs}) on the left hand side,
yields
\begin{equation} 
  \label{eq:main_rel}
  2C_0(t)\Yy_t^0+2C_A(t)\Yy_t^A=\UUumbilic_t+\dot\curv(t)\Yy^{\dimu}_t .
\end{equation}
This restriction links the umbilicity function $\UUumbilic$ and the rate of change of the
space curvature to the functions $\{  C_0, C_A\}$.
The other functions $\{\FF_{0A}, \FF_{AB}\}$ in \eqref{eq:V=} are
associated with the isometries of each leaf, and they remain completely free.

The geometrical meaning of the functions $\{ C_0, C_A\}$
can be determined by evaluating  \eqref{eq:main_rel} as well as
its derivatives along $X_A(t)$,
along the curve  $\oo$. Since in the coordinates $\{ \tau, \lambda^A\}$ this curve is simply $\{\tau =t, \lambda^A=0\}$ and the orthonormal basis is $X_A (t) = \partial_A |_{\oo(t)}$, a
simple computation yields
\begin{equation}
  \label{eq:at_origin}
\UUumbilic|_{\oo(t)}=2C_0(t),\quad X_A(t) (\UUumbilic)|_{\oo(t)}=2C_A(t).
\end{equation}
Furthermore, if we compute the Hessian of \eqref{eq:main_rel}
and use \eqref{eq:Hesst} the following equation
involving only
$\UUumbilic$ and $\curv$ follows
\begin{equation}
  \label{eq:hess_umbilic}
  \Hess_t\UUumbilic_t=-(\curv(t)\UUumbilic_t+\dot\curv(t))\GammaUU_t.
\end{equation}
We are ready to prove a classification result for
umbilic foliations
with constant curvature leaves.
\begin{proposition}[Umbilic foliation]
  \label{Umbilicfoliation}
  Let $(\UU,\intervalI,\tau, \GammaUU)$ be
  an umbilic foliation  with constant curvature leaves (cf.
  Definition \ref{Def:Umbilic})  with umbilicity
  function $\UUumbilic$ and curvature $\curv(t)$.
Assume that $\UU$ has dimension $\dim+1$ with $\dim \geq 3$ and
  let $\oo: \intervalIo\subset\intervalI \rightarrow \UU$
  be any curve parametrized by $t$ and $\{X_A(t)\}$ an orthonormal
  basis of $T_{\oo(t)} \UU$ smoothly depending on $t$.
  
  Then there exists an open neighbourhood $\domainl$ of $\oo$,
  unique coordinates $\{ \tau,\lambda^A\}$ on $\domainl$, and unique
  functions $C_0, C_A,
  \FF_{0A}, \FF_{AB} = - \FF_{BA} \in C^{\infty}(\intervalIo, \mathbb{R})$ such that
  $\oo(t) = \{ \tau=t, \lambda^A =0\}$, $X_A (t) = \partial_A|_{\oo(t)}$ and
  \begin{align}
    \UUumbilic &=
                 2\onman{C}_0 
                 + \left ( 1+\frac{\onman{\curv}}{4}\lsq \right )^{-1} 
                 \left ( 2\onman{C}_A\lambda^A -
                 \left (\onman{C}_0\onman{\curv}+ \frac{1}{2} \onman{\dot\curv}
                 \right  ) \lsq \right ),
                 \label{res:expansion} \\
    \GammaUU &=\GamVV d\tau^2+\conffactor^{-2}
               \delta_{AB} \bm\theta^A\bm\theta^B,\label{eq:gamma_pre_final}\\
    \bm\theta^A &:=d\lambda^A
                  -\left\{\onman{C}_0\lambda^A+\left(\onman{C}_B
                  +\frac{\onman{\curv}}{2}\onman{\FF}_{0B}\right)
                  \left(\lambda^B\lambda^A-\frac{1}{2}\lsq\delta^{AB}\right)+\onman{\FF}_{0B}\delta^{BA}
                  +\onman{\FF}_{BC}\lambda^B\delta^{CA}\right\}d\tau.\nonumber
  \end{align}
  for some function $\GamVV \in C^{\infty}(\domainl, \mathbb{R})$.
\end{proposition}

\begin{proof}
  Another way of writing equation \eqref{eq:main_rel}  is
  \begin{equation*}
    \UUumbilic=2\onman{C}_0\Yy^0 + 2 \onman{C}_A \Yy^A
    -\onman{\dot\curv}\Yy^{\dimu},
  \end{equation*}
  which becomes
  \eqref{res:expansion} after we insert \eqref{eq:Ys}. To show
  \eqref{eq:gamma_pre_final}
  we first note that \eqref{eq:V=} gives
  \begin{equation}
    V  =\partial_\tau+\onman{C}_0 \lambda^B\partial_B+
          \left(\onman{C}_A+\frac{\onman{\curv}}{2}\onman{\FF}_{0A}\right)
          \left(\lambda^A\lambda^B-\frac{1}{2}\lsq \delta^{AB} \right)\partial_B+\onman{\FF}_{0A}\delta^{AB}\partial_B+\onman{\FF}_{AC}\lambda^A\delta^{CB}\partial_B,
          \label{Expr:V}
  \end{equation}
  after recalling that $\W = \partial_{\tau}$ in the coordinates $\{ \partial_{\tau}, \partial_A\}$ and inserting the explicit expressions
  \eqref{ckvl}-\eqref{killslt}. Using this in (\ref{eq:gamma_pre}) brings
  $\GammaUU$ into the form \eqref{eq:gamma_pre_final}. Since the stereographic coordinates $\{\lambda^A\}$ are unique once the curve $\oo$ and the orthonormal
  frame $\{ X_A\}$ along $\oo$ has been chosen, the various uniqueness claims
  hold.
\end{proof}

\begin{remark}
    \label{rem:extension}
The neighbourhood $\domainl$ extends as far as the stereographic coordinates $\lambda^A$ around the curve $\oo$ can be  defined. This property will be used below.
\end{remark}
So far we have found necessary consequences of Definition \ref{Def:Umbilic}. To make sure that no other information can be extracted we need a reciprocal result. 
\begin{proposition}\label{res:ball}
  Let $(\domainl, \intervalIo,\tau, \GammaUU)$
  be defined by $\domainl := \intervalIo \times B_0 (r)$, where $B_0(r)$
  is a ball in $\mathbb{R}^{\dim}$ centered at $0$ with sufficiently small radius $r$, $\tau \in C^{\infty}(
  \domainl, \mathbb{R})$ is the projection of $\intervalIo \times B_0(r)$
  onto the first factor, and $h$ is
  given  by \eqref{eq:gamma_pre_final} where   $\{\lambda^A\}$ are
  Cartesian coordinates on $B_0(r)$. Then
  $(\domainl,\intervalIo,\tau,  \GammaUU)$
is an umbilic foliation with constant curvature leaves with umbilicity function
  given by \eqref{res:expansion} and curvature $\curv(t) := \onman{\curv}_t$.
\end{proposition}

\begin{proof}
It is immediate to check  that
$(\domainl, \intervalIo,\tau,\GammaUU)$ is a metric foliation. Condition 2 in
Definition \ref{Def:Umbilic}
is also immediate. For condition 1 note that the normal vector of the foliation is given by
\eqref{Expr:V}, hence also \eqref{eq:V=} holds and the validity of
\eqref{eq:umbilic_fol} with $\UUumbilic$ given in \eqref{res:expansion}
follows easily from the fact that
$\eta^{\ta\tb}_t$ are conformal Killing vectors satisfying
\eqref{eq:zeta_eqs}.
\end{proof}

\subsection{RW foliation}
\label{RWfoliation}
For our characterization purposes, it is of interest to find sufficient
conditions along the curve $\oo$ that guarantee that the umbilicity function
$\UUumbilic$ is homogeneous, i.e. constant on each leaf $\SigmaUUo_t$. To that aim, observe that the umbilicity function \eqref{res:expansion}
is the solution of the equation
(\ref{eq:hess_umbilic}) with boundary data (\ref{eq:at_origin}).
It is obvious from the explicit form in \eqref{res:expansion} that
$\UUumbilic$ is homogeneous if and only if $C_A =0$ and
$2 C_0 \curv + \dot{\curv}=0$.
We translate this into two geometric conditions along the curve $\oo$. We shall say that Condition RW.1 (resp.  RW.2)  hold whenever

\begin{itemize}
\item[RW.1] The spatial gradient of the
  umbilicity function $\UUumbilic$ vanishes on the curve $\oo$, i.e.
  $X(\UUumbilic) |_{\oo(t)} =0$ for all $t \in \intervalIo$ and for all $X \in T_{\oo(t)}
  \SigmaUU_t$.
\item[RW.2] The spatial Laplacian of the
  umbilicity function $\UUumbilic$ vanishes on $\oo$, i.e.
  $\Delta_t \UUumbilic |_{\oo(t)} =0$ for all $t \in \intervalIo$.
\end{itemize}

From \eqref{eq:at_origin}
it is immediate to check that
Condition RW.1 is equivalent to $C_A=0$. Combining
\eqref{eq:at_origin} with the trace of \eqref{eq:hess_umbilic} at $\oo(t)$
it follows that Condition RW.2 is equivalent to
$2 C_0 \curv + \dot{\curv}=0$. Thus, we have the following lemma.
\begin{lemma}
  \label{RWconds}
  Assume the setup of Proposition \ref{Umbilicfoliation}. The
  umbilicity function $\UUumbilic$ is homogeneous, i.e. there exists
  a function $\UURW\in C^\infty(\intervalIo,\mathbb{R})$
  such that   $\UUumbilic=\UURW\circ\tau=\onman{\UURW}$, if and only if
  conditions RW.1 and RW.2 hold. In such case
  $C_A=0$, $C_0 = \frac{1}{2} \UURW$, and the functions $\curv$
  and $\UURW$ are linked by
  \begin{equation}
    \label{eq:UURW}
    \dot\curv+\UURW\curv=0.
    \end{equation}
\end{lemma}

The following result determines
the form of $\GammaUU$ in the case when
the geometric conditions RW.1 and RW.2 hold.

\begin{proposition}[RW foliation]\label{res:RWfoliation}
  Assume the setup of Proposition \ref{Umbilicfoliation} 
  and, in addition, that
  the geometric conditions RW.1 and RW.2 hold. Then, there exist
  coordinates $\{\tau, z^A\}$ on $\domainl$, a nowhere zero function $a \in C^{\infty}(\intervalIo, \mathbb{R})$, a constant $\curvN$ and
  functions $\FF_{0B}, \FF_{AB} = - \FF_{BA} \in C^{\infty}(\intervalIo,\mathbb{R})$ 
  such that 
    \begin{align}
  &\GammaUU=\GamVV d\tau^2+\onman{a}^2\left(1+\frac{\curvN}{4} \zsq\right)^{-2}
    \delta_{AB} \bm\omega^A\bm\omega^B,  \label{eq:gamma_RW_a}\\
  &\bm\omega^A=dz^A
    -\left\{
    \frac{\curvN}{2\onman{a}}\onman{\FF}_{0B}
    \left(z^Bz^A-\frac{1}{2}\zsq\delta^{AB}\right)+\frac{1}{\onman{a}}\onman{\FF}_{0B}\delta^{BA}
    +\onman{\FF}_{BC}z^B\delta^{CA}\right\}d\tau,\nonumber
    \end{align}
    where $\zsq  := \delta_{AB} z^A z^B$. In addition, the curve $\oo$ is
    given by $\{ \tau =t, z^A=0\}$ and the frame
    $\{X_A(t)\}$ is $X_A(t) = a(t)^{-1} \partial_{z^A}|_{\oo(t)}$.
 \end{proposition}

  \begin{proof}
    By Lemma \ref{RWconds} we have $C_A=0$, $C_0 = \frac{1}{2} \UURW$,
    and \eqref{eq:UURW}
    holds. Inserting into \eqref{eq:gamma_pre_final} yields
  \begin{align}
  &\GammaUU=\GamVV d\tau^2+\conffactor^{-2}
    \delta_{AB} \bm\theta^A\bm\theta^B,\label{eq:gamma_RW_fol}\\
  &\bm\theta^A=d\lambda^A
    -\left\{\frac{1}{2}\onman{\UURW}\lambda^A+
    \frac{\onman{\curv}}{2}\onman{\FF}_{0B}
    \left(\lambda^B\lambda^A-\frac{1}{2}\lsq\delta^{AB}\right)+\onman{\FF}_{0B}\delta^{BA}
    +\onman{\FF}_{BC}\lambda^B\delta^{CA}\right\}d\tau. \label{thetaA}
  \end{align}
  Perform the coordinate change $\lambda^A = \onman{a} z^A$ where
  $a(t)$ is any non-zero solution of the ODE
  \begin{equation}
    \dot{a} = \frac{1}{2} a \UURW. \label{ODE:s}
  \end{equation}
  The function $\curv a^2$ satisfies
  \begin{equation*}
    \frac{d}{dt} \left ( \curv a^2 \right ) = \dot{\curv} a^2 + 2
     \curv a \dot{a} =0
  \end{equation*}
 because of \eqref{ODE:s} and \eqref{eq:UURW}.
  Thus $\curv = \curvN a^{-2}$ for some constant $\curvN$.
  Since
\begin{equation*}
  d\lambda^A = \frac{1}{2}\onman{a} \onman{\UURW} z^A d \tau
  + \onman{a} d z^A
\end{equation*}
the covector $\bm \theta^A$ in \eqref{thetaA} becomes
\begin{equation*}
  \bm \theta^A =
    \onman{a} \left ( d z^A -\left\{
    \frac{\curvN}{2\onman{a}}\onman{\FF}_{0B}
    \left(z^Bz^A-\frac{1}{2}\zsq\delta^{AB}\right)+\frac{1}{\onman{a}}\onman{\FF}_{0B}\delta^{BA}
  +\onman{\FF}_{BC}z^B\delta^{CA}\right\}d\tau \right ).
\end{equation*}
After defining $\bm \omega^A := \onman{a}^{-1} \bm \theta^A$, the form
of the metric \eqref{eq:gamma_RW_fol} follows. Clearly, the
curve $\oo(t)$ is defined by $\{ \tau =t, z^A=0\}$ and the claim on the basis
$\{X_A(t)\}$ follows from
$\partial_{z^A} = \onman{a} \partial_{\lambda^A}$.
\end{proof}

Observe that the form of the normal vector of the foliation $V$ in the coordinates $\{\tau, z^A\}$ can be obtained directly from \eqref{eq:gamma_RW_a} and reads
\begin{equation}\label{eq:V_in_z}
  V = \partial_{\tau} + 
     \left\{\frac{\curvN}{2\onman{a}}\onman{\FF}_{0A}\left(z^Az^B-\frac{1}{2}\zsq \delta^{AB} \right)+
       \frac{1}{\onman{a}}\onman{\FF}_{0A}\delta^{AB}+\onman{\FF}_{AC}z^A\delta^{CB}\right\}
  \partial_{z^B}. 
\end{equation}
Note that $\partial_{\tau}$ in this expression is {\it not} the geometrically defined  vector $\W$ because $\partial_{\tau}$ is a
different vector field in the coordinates $\{ \tau,\lambda^A\}$ than
in the coordinates $\{ \tau, z^A\}$.
Explicitly, in the latter coordinates we have
  \begin{equation}\label{eq:W_in_z}
    \W=\partial_\tau-\frac{\onman{\dot{a}}}{\onman{a}}z^A\partial_{z^A},
  \end{equation}
and $\W|_{\oo(t)}=\partial_\tau|_{\oo(t)}$.
 
An umbilic foliation $(\domainl,\intervalIo,\tau, \GammaUU)$ for which Proposition
  \ref{res:RWfoliation} holds will be called a {\bf RW foliation.}

\section{ New characterization of RW geometries}
\label{NewCharac}
In the previous section the tensor $\GammaUU$ was allowed to be degenerate, Lorentzian or Riemannian depending on the point. This is why we could say nothing about the function $\GamVV$ that appears in \eqref{eq:gamma_pre_final}. Indeed,
this function is related to
$V$ by means of
\begin{equation}
  \GamVV = \GammaUU(V,V)
\end{equation} 
(this follows immediately from $\bm \theta^A (V) = 0$, cf. \eqref{eq:gamma_pre_final} and \eqref{Expr:V}, and $d \tau(V)=1$). Thus,  in order to restrict $\GamVV$ we need to make further assumptions on $\GammaUU$. 

In this section we assume that $\GammaUU$ is a metric. To emphasize this fact 
we replace $\GammaUU$ by $g$ from now on.
Although our main interest is when $g$ is Lorentzian we still allow for the Riemannian case as this entails no extra effort.
$g$ being a metric is equivalent to $g(V,V)$ being nowhere zero.
Thus, there exists a positive function $H \in C^{\infty} (\UU,\mathbb{R})$
and a sign $\signature \in \{ -1, 1\}$ such that
\begin{equation}
  \label{GamVV2}
  g(V,V) = - \signature H^2.
\end{equation}
The space $(\UU,g)$ is a semi-Riemannian
manifold with signature   $\{-\signature,1 ,\cdots, 1\}$. Since $\SigmaUU_t$ are now  spacelike
hypersurfaces of $(\UU,g)$, they admit a unique unit normal $\flow$ pointing along increasing values of $\tau$, i.e. satisfying
\begin{equation*}
  g (\flow,\flow) = -\signature, \qquad \flow(\tau) >0.
\end{equation*}
By construction $V$ and $\flow$ are proportional to each other (because both are $g$-orthogonal to $T_p \SigmaUU_\tau(p)$ for all $p \in \UU$).  Combining with
\eqref{GamVV2} it follows that
\begin{equation}
  V = H \flow. \label{Vflow}
\end{equation}
For any vector $X$ we denote by $\bm{X}$ the covector obtained by lowering indices with $g$. It follows directly from its definition that
$\bm{V}$ is proportional to $d\tau$ and, in fact,
\begin{equation}
  \bm{V} = - \signature H^2 d\tau \label{Vtau}
\end{equation}
because writing $\bm{V} = P d \tau$ we have
\begin{equation*}
  1=  V(\tau) = d\tau (V) =  P^{-1} \bm{V} (V) = P^{-1} g(V,V) = - \signature P^{-1} H^2 \qquad \Longleftrightarrow \quad P = - \signature H^2.
\end{equation*}
Consequently,
\begin{equation}
  \bm{\flow} = - \signature H d\tau. \label{udtau}
\end{equation}
The hypersurface  $\SigmaUU_t$ admits a second fundamental form
$K_t$ which can be computed using
\begin{equation}\label{secFFpre}
  K_t = \frac{1}{2} i_t^{\star} \left ( \pounds_\flow g \right ).
\end{equation}
Thus,
\begin{equation}
  K_t = \frac{1}{2} i_t^{\star} \left ( \pounds_{H^{-1} V} g \right )
  = \frac{1}{2 H_t} i_t^{\star} \left ( \pounds_{V} g \right ),
  \label{secFF}
\end{equation}
where we applied the straightforward identity
\begin{equation*}
  \pounds_{f X} g = f \pounds_X g + df \otimes \bm{X} +
  \bm{X} \otimes df
\end{equation*}
together with \eqref{Vflow}. By \eqref{secFF}, the notion of $(\UU,\intervalI,\tau,g)$ being an umbilic foliation, cf. \eqref{eq:umbilic_fol}, is equivalent to the second fundamental form $K_t$ of $\SigmaUU_t$ being pure trace (also called {\it totally umbilic} or {\it shear-free}). In such case, the second fundamental form
reads
\begin{equation}
  K_t = \frac{\UUumbilic_t}{2H_t} g_t. \label{umbilic_Expansion}
\end{equation}

Note that the class of spacetimes given by $(\UU,g)$ produces a class of cosmological models
with spatial constant curvature
$\curv(t)$ that can change sign from leaf to leaf.
Observe that in a RW cosmology the sign of $\curv(t)$
cannot change. The class $(\UU,g)$ thus belongs to the ``SI''
(space-isotropic) type in the classification produced in \cite{Avalos2022},
where the original family of spacetimes recently
constructed in \cite{Sanchez2023}
is analysed as a paradigmatic example.

Now, recall the expansion (also called mean curvature)
of a hypersurface is the trace of its second fundamental form. 
It is of interest to find the most general form of $g$
when not only $K_t$ is pure trace, but in addition the expansion is homogeneous and non-zero, i.e. of the form
$\onman{\expanshom}$
for some smooth nowhere zero function
$\expanshom \in C^{\infty}(\intervalI, \mathbb{R})$.
\begin{proposition}[Homogeneous expansion]\label{res:non_geodesic}
  Let $(\UU,g)$ be a semi-Riemannian manifold of dimension $\dim +1$, $\dim \geq 3$,  
  and signature $\{ - \signature, +1, \cdots, +1\}$
  endowed with a function
  $\tau \in C^{\infty} (\UU,\intervalI)$,
  for some interval $\intervalI \subset \mathbb{R}$,
  satisfying $d\tau \neq 0$ everywhere.
  Let $\SigmaUU_t$ be the level sets of $\tau$ and assume that:
  \begin{itemize}
  \item[(i)] The induced metric $g_t$ on each $\SigmaUU_t$ is of constant curvature $\curv(t)$.
  \item[(ii)] There is a nowhere zero function $\expanshom \in C^{\infty} (\intervalI,
    \mathbb{R})$ such that the second fundamental form $K_t$ of
    $\SigmaUU_t$  is of the form
    \begin{equation}
      K_t = \frac{\expanshom(t)}{\dim} g_t. \label{Expansion}
    \end{equation}
  \end{itemize}
  Select any curve $\oo$ parametrized by $t$ and
  an orthonormal frame $\{ X_A(t)\}$ of $T_{\oo(t)} \SigmaUU_{t}$ smoothly depending on $t$. Then, there is a neighbourhood $\domainl$ of $\oo$ and
  unique coordinates $\{ \tau,\lambda^A\}$ on $\domainl$ such that the metric $g$
  takes the form \eqref{eq:gamma_pre_final} with $\GamVV$ given by
  \begin{equation}
    \label{expressionA}
    \GamVV = - \frac{\signature \dim^2}{\onman{\expanshom}^2}
    \left ( 
    \onman{C}_0 
    + \left ( 1+\frac{\curv}{4}\lsq \right )^{-1} 
    \left ( \onman{C}_A\lambda^A -
    \left (\onman{C}_0\onman{\curv}+ \frac{1}{2} \onman{\dot\curv}
    \right  ) \frac{\lsq}{2} \right ) \right )^2.
  \end{equation}
\end{proposition}

\begin{proof}
  All the hypothesis of Proposition \ref{Umbilicfoliation} hold.
  Comparing \eqref{Expansion} with \eqref{umbilic_Expansion} we get
  $H_t = \dim \UUumbilic_t/(2 \expanshom(t))$ or equivalently
  \begin{equation*}
    H = \frac{\dim \UUumbilic}{2 \onman{\expanshom}}.
  \end{equation*}
  Since $\GamVV= - \signature H^2$ and $\UUumbilic$ is given by
  \eqref{res:expansion}
  the result follows.
\end{proof}

We are ready to state and prove our main results of the paper. The only extra restriction we need is that the unit normal vector $\flow$ is geodesic.
In the first theorem we find necessary and sufficient conditions for any
metric admitting a foliation by
totally umbilic hypersurfaces of constant curvature and with geodesic unit normal to be locally isometric to a RW space. Again we allow for both Lorentzian and Riemannian signatures. In the second theorem we write down  the RW metric in canonical coordinates adapted to an arbitrary transverse curve carrying an orthonormal frame
adapted to the foliation.

\begin{theorem}[Local characterization of RW in terms of a curve]
  \label{main1}
  Let $(\UU,g)$ be a semi-Riemannian manifold of dimension $\dim +1$, $\dim \geq 3$,
  and signature $\{ - \signature, +1, \cdots, +1\}$
  endowed with a function
  $\tau \in C^{\infty} (\UU,\intervalI)$,
  for some interval $\intervalI \subset \mathbb{R}$, satisfying $d\tau \neq 0$ everywhere.
  Let $\SigmaUU_t$ be the level sets of $\tau$ and assume that each
  $\SigmaUU_t$ is connected and that:
  \begin{itemize}
  \item[(i)] The induced metric $g_t$ on each $\SigmaUU_t$ is of constant curvature $\curv(t)$.
  \item[(ii)] The second fundamental form of $\SigmaUU_t$ is pure trace, i.e. there is a function $\expansion \in C^{\infty}(\UU,\mathbb{R})$ such that
    \begin{equation}
      K_t = \frac{\expansion_t}{\dim} g_t. \label{Expansion2}
    \end{equation}
  \item[(iii)] The unit normal $\flow$ to $\SigmaUU_t$ is geodesic.
  \end{itemize}
  Then, $(\UU,g)$ is locally a RW space if and only if there exists
     a curve $\oo$ parametrized by $t$ and defined all over $\intervalI$ such that
  the expansion $\expansion$
    along the curve satisfies
  \begin{equation}
   \grad_{g_t}\expansion_t |_{\oo(t)}=0, \qquad \Delta_t \expansion_t |_{\oo(t)} =0.
    \label{conditions}
  \end{equation}
\end{theorem}

\begin{proof}
  We start with the  ``only if'' part, i.e. we assume that $(\UU,g)$
    is locally isometric to  RW. The expansion
    $\expansion$ in a RW space is constant on each orbit of the isometry group, so
    for any transversal curve $\oo$ conditions \eqref{conditions} hold.

    For the ``if'' part, we denote by $\nablag_i$ the  Levi-Civita derivative of $g$ and use
  abstract indices $i,j,\ldots$ on $\UU$. By items (i) and (ii),
  $(\UU,\intervalI,\tau,g)$ is an umbilic foliation with constant curvature leaves (cf. Definition \ref{Def:Umbilic})  and the unit normal $\flow$ to
    the leaves $\SigmaUU_t$ is given by \eqref{udtau}. We compute the acceleration of $\flow$ as follows (see e.g. \cite{Gourgoulhon})
    \begin{align}
    (\nablag_\flow\flow)_i =
    &
      - \signature \flow(H)\nablag_i\tau-
      \signature H \flow^j\nablag_j\nablag_i\tau
      = \frac{\flow(H)}{H} \flow_i
      - \signature H \nablag_i \left ( \flow^j \nablag_j \tau \right )
      -\flow_j\nablag_i \flow^j
            \nonumber \\
    &       =  \frac{\flow(H)}{H} \flow_i
      - \signature H \nablag_i \left ( H^{-1} \right )
      = \frac{\flow(H)}{H} \flow_i + \signature \frac{1}{H} \nablag_i H, \label{b}
    \end{align}
    where we have used $\flow$ is unit in the third equality.
  It follows from \eqref{b} that
  $\flow$ is geodesic if and only if there exists a function
  $T \in C^{\infty} (\intervalI,\mathbb{R} )$ such that
  $H = \onman{T}$. From \eqref{Expansion2}  and \eqref{umbilic_Expansion}
  we have
  \begin{equation*}
    \UUumbilic = \frac{2}{\dim} \onman{T} \expansion.
  \end{equation*}
  Conditions \eqref{conditions} are therefore equivalent to RW.1 and RW.2
  in subsection  \ref{RWfoliation}. Therefore, by Proposition \ref{Umbilicfoliation}
    and Lemma \ref{RWconds} 
    there is a neighbourhood $\domainl$ of $\oo$ 
    so that $(\domainl,\intervalI,\tau,g)$ is
    a RW foliation and the
  umbilicity function $\UUumbilic$
  is homogeneous on $\Sigma_t$, that is,
  its gradient is pointwise parallel to $\flow$.
  We can thus apply the local result of Theorem \ref{Miguel}
  to conclude that $\domainl$ with the restricted metric
  is locally a RW space. It only remains to show that this claim holds for $\domainl = \UU$.

We can now take any other curve $\oo'$ parametrized by $t$,
  and defined all over $\intervalI$ or any subinterval thereof, in the
  domain $\domainl$.
  By the ``only if'' part of the theorem, conditions \eqref{conditions} hold also for
  $\oo'$, and we can thus construct its corresponding domain $\domainl'$.
  
  By homogeneity of the leaves, the construction
  of the stereographic coordinates $\{z^A\}$
  in Proposition \ref{res:RWfoliation}
  (see Lemma \ref{properties}) 
  is independent of the curve. This implies, in particular, that for any $t\in \intervalI$
  there exists a sufficiently small $r(t) >0$ such that the geodesic ball ${\mathcal B}_p (r(t))$ of radius $r(t)$ centered at $ p \in \Sigma_t$ is covered by the stereographic coordinates
  $\lambda^A$ with origin at $p$. Note that $r(t)$ is independent of $p \in \Sigma_t$. By
  Remark \ref{rem:extension} the neighbourhood $\domainl'$ contains
  ${\mathcal B}_{\oo'(t)} (r(t))$.  Since each $\SigmaUU_t$
  is connected, it is now clear that 
  we can reach any point by overlapping
  domains $\domainl$ for different curves, and therefore
  cover the whole of $\UU$.
  As a result the unit normal $\flow$ to $\Sigma_t$ satisfies all conditions
 (i), (ii), (iii)
  of Theorem \ref{Miguel}  and the result follows.
  \end{proof}

  The conclusion of the theorem, namely that $(\UU,g)$
        is locally a RW space, cannot be strengthed because our global\footnote{The term `''global'' is used with two different meanings in this paper, namely  local conditions that hold everywhere and conditions of truly global nature. The first meaning is only used in the Introduction, so there is no room for misunderstanding.} 
    assumptions on $(\UU,g)$ are too weak to conclude more.
 Indeed, one can remove any open set from $(\UU,g)$ away from the curve $\oo$ and all the hypothesis of the theorem still hold. This prevents
 $(\UU,g)$ from being globally isometric to a RW space. To get a global result we need additional hypotheses. We write down one such result based on the
 global characterization in 
 \cite{Caballero2011}, quoted in Theorem \ref{Miguel} above.

\begin{theorem}[Global characterization of RW in terms of a curve]
  \label{main1_global}
  Let $(\UU,g)$ satisfy the hypotheses of Theorem \ref{main1}
    and let $\flow$ be the vector field defined
  in its item (iii).
  Assume  that for some $t_0 \in \mathbb{R}$ there exists
    an interval $I \subset \mathbb{R}$ such that the map $\Phi: I \times
    \Sigma_{t_0}
    \to M$ defined to be the flow of the vector field
    $\flow$ is well-defined and onto.
Then $(\UU,g)$ is a RW space if and only if there exists a curve
$\oo: \intervalI \to \UU$ parametrized by $t$
such that 
the expansion $\expansion$ satisfies
  \eqref{conditions} for all $t\in\intervalI$.
\end{theorem}
\begin{proof}
  In the proof of  Theorem \ref{main1} we have shown that
  $(\UU,g)$ satisfies conditions (i) to (iii) of
  Theorem \ref{Miguel}. Thus, the result follows from the global part of
  Theorem \ref{Miguel}.
\end{proof}

\begin{remark}
  The purely local conditions that we are imposing along the curve $\oo$ involve the expansion $\expansion$ of $\flow$, cf. equations \eqref{conditions}. It is of interest to translate these conditions in terms of the energy-momentum contents. So, let us assume that we are in the context
  of Theorem \ref{main1} so that $(\UU,g)$ satisfies conditions 
(i) to (iii) and, only for this remark,  that the Einstein equations hold. The momentum constraint on the
  hypersurface $\SigmaUU_t$ is
  \begin{align*}
\mbox{div}_{g_t} \left ( K_t - \expansion_t g_t \right ) = - J_t, 
  \end{align*}
  where $\mbox{div}_{g_t}$ is the divergence in $(\SigmaUU_t, g_t)$ and
  $J_t$ is the energy-flux with respect to the observer $\flow$ at
  $\SigmaUU_t$. Since
  the second fundamental form of $\SigmaUU_t$ is pure trace (cf.
  \eqref{Expansion2}), this equation becomes
  \begin{align*}
    \frac{\dim -1}{\dim} \grad_{g_t} \expansion_t= J_t.
  \end{align*}
  Thus, conditions \eqref{conditions} can be equivalently stated in terms of the
  energy flux by imposing that both $J_t$ and its divergence
  $\mbox{div}_{g_t} J_t$ vanish along the curve $\oo$.
\end{remark}

An interesting by-product of the characterization of RW spaces is the construction of an explicit coordinate system canonically
  adapted to any transverse curve and to any orthonormal space-frame defined along the curve. This frame spans, at every instant of cosmological time,  the corresponding cosmological rest space. Note that (in the Lorentzian setting) the curve is allowed to be of any causal character, so in general there is no local rest space for the curve itself.

This coordinate system could be relevant to study cosmological effects felt by a
  single observer moving arbitrarily with respect to the cosmological flow. Obviously, for such case it is necessary to restrict the transverse curve to be timelike.

\begin{theorem}[Coordinates in RW defined along an arbitrary transverse curve]
  \label{main2}
  Consider a RW space $(M = \intervalI \times \Sigma,g)$
  of dimension $\dim +1$, $\dim \geq 3$, with scale factor $a:\intervalI\to \mathbb{R}$ and curvature $\curvN$, i.e.
  
  \begin{equation*}
    g =  -\signature d\tau^2 + a^2(\tau) g_{\curvN}, \qquad \tau \in \intervalI,
  \end{equation*}
  where $g_{\curvN}$ is a metric of constant curvature $\curvN$
  on $\Sigma$.
  Let $\Sigma_t$ be the hypersurfaces $\{ \tau = t\}$.    Select
  any smooth curve $\oo: \intervalI \rightarrow M$ satisfying $\oo(t) \in
  \Sigma_t$ and an orthonormal frame $\{ X_A (t)\}$ of $T_{\oo(t)} \Sigma_t$. Then there exist coordinates $\{ \tau, z^A\}$ and functions
  $ \FF_{0A},  \FF_{AB} = - \FF_{BA} \in C^{\infty} (\intervalI, \mathbb{R})$ such that
  \begin{align}
    &g= - \signature d\tau^2+a^2(\tau)\left(1+\frac{\curvN}{4}\zsq\right)^{-2}
      \delta_{AB} \bm\omega^A\bm\omega^B, \label{corrg}  \\
    &\bm\omega^A=dz^A
      -\left\{
      \frac{\curvN}{2a(\tau)}{\FF}_{0B}(\tau)
      \left(z^Bz^A-\frac{1}{2}\zsq\delta^{AB}\right)+\frac{1}{a(\tau)}\FF_{0B}(\tau)\delta^{BA}
      +\FF_{BC}(\tau)z^B\delta^{CA}\right\}d\tau,\label{bomegaA}
  \end{align}
  where $\zsq  := \delta_{AB} z^A z^B$. In addition, the curve $\oo$ is
  given by $\{ \tau =t, z^A=0\}$ and the frame
  $\{X_A(t)\}$ is $X_A(t) = a(t)^{-1} \partial_{z^A}|_{\oo(t)}$.
\end{theorem}

\begin{proof} The set
  $(M,\intervalI,\tau,g)$ is clearly a RW foliation and the second fundamental form of the constant curvature leaves $\Sigma_t =\{ \tau = t\}$, $t\in \intervalI$, satisfies
    \begin{align}\label{Kt_final}
      K_{t} = \frac{\dot{a}(t)}{a(t)}g_t.
    \end{align}
  Proposition \ref{res:RWfoliation}
  provides the form of
  $g=h$ and $\bm{\omega}^A$, leaving the function $\GamVV$ to be determined.
  The proof of this proposition also gives  $C_A=0$, $C_0 = \dot{a}/a$ and $\curv = \curvN a^{-2}$.
  Combining \eqref{Kt_final} with \eqref{Expansion} gives  $\expansion_t = \dim \dot{a}(t)/a(t)$.
We can now apply expression \eqref{expressionA} in Proposition \ref{res:non_geodesic} to conclude $\GamVV = -
   \signature$ and the theorem is proved.
\end{proof}

\begin{remark}\label{main2:u}
Since $\GamVV = - \epsilon$ the foliation function $\tau$ is such that $d\tau$ is unit. Thus $H=1$ as a consequence of \eqref{udtau} and hence $V = \flow$ by \eqref{Vflow}. Consequently the flow vector $\flow$ in the coordinates of the theorem takes the form \eqref{eq:V_in_z}
  \begin{equation}\label{eq:u_e}
    \flow= \partial_{\tau} + 
    \left\{
      \frac{\curvN}{2a(\tau)}\FF_{0A}(\tau)\left(z^Az^B-\frac{1}{2}\zsq \delta^{AB} \right)
      +\frac{1}{a(\tau)}\FF_{0A}(\tau)\delta^{AB}+\FF_{AC}(\tau)z^A\delta^{CB}\right\}
    \partial_{z^B}. 
  \end{equation}
\end{remark}

\begin{remark}\label{velocity}
    In the Lorentzian case $\signature=1$,
    the curve $\oo(t)$ is timelike iff $\FF_{0A}(t)\FF_{0B}(t)
  \delta^{AB}<1$, and then it describes the single observer that moves
    with three-velocity
    \[
      v(t)=-a^{-1}(t)F_{0B}(t)\delta^{AB}\partial_{z^B}
    \]
    with respect to the cosmological flow $\flow$.
    On the other hand, as expected, the functions $F_{AB}$ relate to the
    rotation of the frame $\{X^A(t)\}$ along the curve
    with respect to the cosmological
    flow.
    We devote Subsection \ref{Sub:Killings} below
    to prove these statements and provide the explicit relations,
    in particular.
  \end{remark}

\begin{remark}
  Theorem \ref{main2} requires $\dim \geq 3$ because the core of our
  argument relies on using a finite dimensional basis of conformal Killing vectors, and this requires that the spaces have at least dimension three. However, one checks that
  when $\dim =1,2$ the metric \eqref{corrg}-\eqref{bomegaA} is still locally isometric to a RW space.
  This can be done in several different ways. A simple one is  to note  that $g$ admits $\dim (\dim+1)/2$ linearly independent Killing vector fields tangent to spacelike hypersurfaces. The explicit form of these Killing vectors is obtained in
  Subsection \ref{Sub:Killings} below.
\end{remark}

\subsection{Description of the new coordinates}
\label{Description}
In this subsection we present a description of the geometric meaning of the
free smooth functions $ \FF_{0A}(t)$ and $\FF_{AB}(t)$. Before doing so, however, it is worth to comment on the method that we have followed to determine the RW metric in the new coordinates.

In Section \ref{sec:geo_H} we have applied the construction in
  Appendix \ref{sec:constant_curv_spaces} to build, in the whole foliated
manifold, stereographic coordinates on each constant curvature leaf
centered on an arbitrary curve $\oo(t)$ crossing the foliation. The
curvature on each leaf $\curvS(t)$ was allowed to change sign. Since
the transformation that sends stereographic coordinates centered at
one point to stereographic coordinates centered at another point in a
space of constant curvarture $\curvS$ is an isometry, one could think
of approaching the problem with a direct coordinate change, namely a
transformation that sends leaves to leaves and that, restricted to
each leaf, defines an isometry (with parameters depending on
$t$). However, this approach has some disadvantages. First, the isometries
take different forms depending on whether $\curvS$ is zero or
non-zero. In the former case the isometries are just the standard
Euclidean motions  and in the second one they are
a subset of the M\"obius transformations \cite{moebius},
so both cases cannot be treated in one go.
Working on a sufficient local patch, the two cases can be treated
simultaneously by means of rotations and so-called
``quasitranslations'' (see e.g. \cite{Weinberg}).  However, such
coordinates do not cover the whole space when $\curvS$ is positive
(this is related to the last comment in Remark \ref{embedding} in
Appendix\ref{sec:constant_curv_spaces}).
So these methods cannot be applied to cases when the curvature is
allowed to change sign and one wants to cover the whole manifold
(up to the antipodal point, obviously).
Even restricting oneself to foliations with everywhere positive,
zero or negative curvature, the direct coordinate transformation
method turns out to be
remarkably more complicated that the approach that we have
followed here. Without knowing beforehand the final form of the tensor
\eqref{eq:gamma_RW_a} (or of the metric \eqref{corrg} in the RW case)
it would have been quite difficult to show that the transformation that
maps leaves into leaves and
sends the origin into a moving point $\oo(\tau)$ takes the final simple
form that we find. In any case, in Remark \ref{rem:curve_F} below we provide
the relation of the parametrization of the curve $\oo(\tau)$
expressed in standard RW comoving coordinates with the coefficients
$F_{0A}$ and $F_{AB}$ in the simplest possible case
(corresponding to a certain rotation of the coordinates being set to
the identity).

We now proceed with the description of $\FF_{0A}$ in terms of the curve $\oo$ and the orthonormal frame $\{X_A(t)\}$  selected along the curve.
First, observe that the character (timelike, spacelike or null)
of the curve $\oo$ is completely arbitrary and may even depend on the point (of course  this only
makes sense for $\signature=1$).
Since the vector tangent to $\oo$ is $\W|_{\oo(t)}$,
  and in the coordinates  $\{\tau,z^A\}$ of Theorem \ref{main2} one has
  $\W|_{\oo(t)} =\partial_\tau|_{\oo(t)}$
and $z^A=0$ at $\oo$, it is straightforward to compute
\[ g(\partial_\tau,\partial_\tau)|_{\oo(t)}=
  -\signature+\FF_{0A}(t)\FF_{0B}(t) 
  \delta^{AB}
  =:N(t).
\]
Therefore, in the case $\signature =1$ the curve $\oo$ is timelike at points where
$\sum_A (\FF_{0A}(t))^2<1$. On the other hand, since $X_A(t) = a(t)^{-1} \partial_{z^A}|_{\oo(t)}$,
we have
\[
  \FF_{0A}(t) =-g(X_A,\partial_\tau)|_{\oo(t)}.
\]
Thus, $\FF_{0A}(t)$ provide a measure of the angle
between the leaves $\SigmaUU_t$ and the curve $\oo$ at time $t$, and thus,
how the projection of the curve $\oo$ on the leaves moves
with $t$. Explicitly, the three-velocity of the observer following
  the curve $\oo$ with respect to the cosmological flow,
  defined as the standard of
  the projection of $\partial_\tau$ orthogonal
  to $\flow$ at $\oo$, is thus explicitly given by
    \begin{align*}
      v=\frac{1}{- g ( \partial_\tau |_\oo, U|_{\oo})}
      \left ( \partial_\tau|_{\oo}+ g(\partial_\tau|_\oo,\flow|_\oo) \flow|_{\oo}
 \right )     =-a^{-1}F_{0B}\delta^{AB}\partial_{z^A},
\end{align*}
      as stated in Remark \ref{velocity}.

Assume now that $\oo$ is nowhere null.
Recall that the Fermi-Walker derivative of a vector field $X(t)$ along a curve $\gamma(t)$ with tangent vector $\dot{\gamma}$ is defined
  as
  \[
    D^{FW}_{\gamma} X := \nabla_{\dot{\gamma}} X +\frac{1}{g(\dot\gamma,\dot\gamma)} \left(g(X, \nabla_{\dot{\gamma}} \dot{\gamma}) \dot{\gamma} - g (X, \dot{\gamma}) \nabla_{\dot{\gamma}} \dot{\gamma}\right),
  \]
  and that it describes the rate of rotation of the field $X(t)$ along the curve.  It is therefore of interest to compute  $D^{FW}_\oo X$
along the curve $\oo$. The computation is aided by the fact that the tangent vector of $\oo$ is the restriction to $\oo$ of the vector field 
$\W$.
We start with the Fermi-Walter derivative of the flow vector $\flow$.
A straightforward calculation yields 
\begin{equation}
  Z:=D^{FW}_\oo \flow=
  -\frac{\signature}{a N}\left(\frac{d}{dt}\FF_{0}{}^A-\FF_{0B} \FF^{BA}\right)\partial_{z^A},\label{eq:Z} 
\end{equation}
where here and in the following we use $\delta^{AB}$ to raise the indexes $A,B,\ldots$. The Fermi-Walker derivative of the frame vectors
$X_A(t)$ is most easily expressed in terms of the functions
 $Z(z^A)$, namely the $A$-components of the vector $Z$ in the basis $\{ \partial_{\tau}, \partial_{z^A}\}$. The result is
\begin{equation}
  D^{FW}_\oo X_A=\left(\signature \FF_{0A} Z(z^B)
    -\frac{1}{a}\FF_{A}{}^B\right)\partial_{z^B}
    +\signature a Z(z^A) \partial_{\tau}. \label{eq:DXA} 
    \end{equation}
Now, if the frame $\{X_A(t)\}$
is transported Fermi-Walker along $\oo$, i.e. $D^{FW}_\oo X_A=0$, then the $\partial_{\tau}$ component of  \eqref{eq:DXA}
yields  $Z(z^A)=0$ and then its $\partial_{z^A}$ component gives $\FF_{AB}=0$, which by the explicit expression of $Z(z^A)$ (cf. \eqref{eq:Z}) also imply that $ \FF_{0A}(t)$ are constant.
Conversely, if $ \FF_{0A}(t)$ are constant and $\FF_{AB}=0$
then  $Z(z^A)=0$ and thus $D^{FW}_\oo X_A=0$ for all $A$.

Observe that demanding that the frame $\{X_A(t)\}$ is
transported Fermi-Walker along $\oo$ restricts
the curve itself to be at a constant angle $ \FF_{0A}(t)$.
This may seem strange at first sight because one could think that the freedom in choosing the frame should allow for any type of Fermi-Walker transport along any curve. However, this is not the case because by construction the frame
$\{X_A(t)\}$  has already been restricted to be tangent to the leaves, which prevents one from choosing an arbitrary Fermi-Walker transport law along
 $\oo$.

It is also of interest to study the consequences of imposing that the flow vector $\flow$ be
transported  Fermi-Walker along $\oo$,
i.e. $Z=0$. Again, this restricts the form of $\oo$. Inserting $Z=0$  in  \eqref{eq:Z}
  gives
\[
 \frac{d}{dt}(\FF_{0A})=\FF_{0}{}^B \FF _{BA}. 
\]
In any case,
since the Fermi-Walker transport of a frame
is associated to its rotation along the curve,
we have thus clear indication that $\FF_{AB}$ is related to  rotation.
We confirm this by showing
  that we can get rid of $\FF_{AB}(t)$
  by means of a rotation of the frame $\{X^A(t)\}$ along {\it any} curve $\oo(t)$.

A rotation of the frame is
a transformation of the form $\newX_A(t)=R_A{}^B(t) X_B(t)$
along $\oo(t)$, where $R_A{}^B(t)$ are a set of $\dim^2$ functions on $\intervalI$
satisfying $R_A{}^C(t) R_B{}^D(t)\delta_{CD}=\delta_{AB}$.
Let us first show
how a rotation translates to a change of the corresponding stereographic coordinates.
Since the coordinates $z'^A$ are adapted to $X'_A(t)$, the transformation above
is equivalent to 
$\partial_{z'^A}|_{\oo(t)}=R_A{}^B(t) \partial_{z^B}|_{\oo(t)}.$
By uniqueness of the stereographic coordinates, see Remark \ref{res:stereo_cords}, the change of coordinates must satisfy
$$\partial_{z'^A}|_{\SigmaUU_t}=R_A{}^B(t) \partial_{z^B}|_{\SigmaUU_t},$$
everywhere on each leaf $\SigmaUU_t$. Hence,
$$
\frac{\partial z^B}{\partial z'^A}=\onman{R}_A{}^B,
$$
and, given that  $\oo$ is located at $z'^A=0$,
the coordinate change on $M$  is given by 
\begin{equation}
  z^B=\onman{R}_A{}^Bz'{}^A \qquad \Longleftrightarrow \qquad
z'^A = \onman{R}^A{}_B z^B(\equiv \delta^{AC}\delta_{BD}\onman{R}_C{}^D z^B).
\label{eq:trans_z}
\end{equation}
In short, a time dependent rotation of the frame $\{X^A(t)\}$
along $\oo(t)$ 
produces a uniform rotation of the coordinates on each leaf $\SigmaUU_t$.
The following result shows how the rotations $\onman{R}_A{}^B$ are related
to $\onman{\FF}_{AB}$.

\begin{lemma}\label{res:rotations_and_F}
  Given the RW space of Theorem \ref{main2} and
  its flow $\flow$ from Remark \ref{main2:u},
  we have
  \begin{equation}
    \label{eq:lie_uX}
    [\flow,\partial_{z^A}]= -\onman{\FF}_{A}{}^{B}\partial_{z^B}.
  \end{equation}
  There exist stereographic coordinates $\{z'^A\}$
    associated to a  new orthonormal
    frame $\{\newX_A(t)\}$ along the curve $\oo(t)$
  for which
  \begin{equation}
    \label{eq:comm_VX}
    [\flow,\partial_{z'^A}]=0.
  \end{equation}
  The frames $\{\newX_A(t)\}$ and $\{ X_A(t)\}$ are related by
    a rotation $\newX_A(t)=R_A{}^B(t)X_B(t)$, which, in matrix form, takes the explicit form
      $R(t)= R(t_0) \exp(-\int^t_{t_0}\FF(s)ds)$,
      where $R(t_0)$ is an arbitrary element
        $R(t_0)\in O(n)$.
 \end{lemma}
\begin{proof}
  Direct computation using \eqref{eq:u_e} yields
  \[
    [\flow,\partial_{z^A}](\tau)=0, \quad
    [\flow,\partial_{z^A}](z^B)=-\onman{\FF}_{A}{}^{B},
  \]
  and so \eqref{eq:lie_uX} follows.
  Now, consider a set of $\dim^2$ functions
  $R_A{}^B\in C^{\infty}(\intervalI,\mathbb{R})$
  that solve the ODE $\dot R_A{}^B(t))=-R_A{}^C(t) \FF_{C}{}^B(t)$ for all $t\in\intervalI$,
  and thus
  \begin{equation}\label{eq:R}
  \flow(\onman{R}_A{}^B)=-\onman{R}_A{}^C\onman{\FF}_C{}^B.
  \end{equation}
  This is a linear system of ODE, so unique solutions exist globally given initial data $R_A{}^B(t_0)$. It is a standard fact that $R_A{}^B(t)$
    are  rotations for all $t$ provided the initial data is a rotation. 
    This follows from the fact that
    $I_{AB}\defi R_A{}^C(t) R_B{}^D(t)\delta_{CD}$ is initially $\delta_{AB}$ and
    \begin{align*}
\dot{I}_{AB}  
      =&-\delta_{CD}\left(R_A{}^E\FF_{EF}\delta^{FC}R_B{}^D
         +R_B{}^E\FF_{EF}\delta^{FD}R_A{}^C\right)\\
      =&-\left(R_A{}^E\FF_{EF}R_B{}^F
         +R_B{}^E\FF_{EF}R_A{}^F\right)=-R_A{}^ER_B{}^F\left(\FF_{EF}+\FF_{FE}\right)=0.
    \end{align*}

  It only remains to define the vector fields
  $\newX_A(t)=R_A{}^B(t) X_B(t)$ along $\oo(t)$.
By construction  $\{\newX_A(t)\}$
 is  an orthonormal frame at each $t\in\intervalI$. 
  As shown above, the corresponding stereographic coordinates satisfy
  \begin{equation}\label{eq:trans_partial_z}
    \partial_{z'^A}=\onman{R}_A{}^B \partial_{z^B}.
  \end{equation}
  Equation \eqref{eq:comm_VX} follows at once from  this together with
    \eqref{eq:lie_uX}
  and \eqref{eq:R}.
\end{proof}

The first result of the lemma tells us first
how to characterize $\FF_{AB}$ in terms of the flow vector $\flow$
and the natural basis $\partial_{z^A}$, explicitly,
\[
  \FF_{A}{}^B(t)=
  - [\flow,\partial_{z^A}](z^B).
\]
The second result ensures that rotating conveniently 
the frame along $\oo(t)$ 
we can get rid of 
$\FF_{AB}(t)$ in the corresponding new stereographic coordinates.

Finally, let us show how the vector field $\W$ transforms under the transformation
$\newX_A(t)=R_A{}^B(t) X_B(t)$ along $\oo(t)$. 
Recall that $\W$ is the vector field tangent to the curves of
  constant values of $\{\lambda^A\}$ and parametrized by $\tau$. In the
coordinates
$\{\tau,z^A\}$, $\W$ is given by \eqref{eq:W_in_z}.
We already
know that the vector fields $\partial_{z^A}$ transform as \eqref{eq:trans_partial_z},
so it is convenient to write \eqref{eq:W_in_z} as
\[
  \W=\flow-\flow(z^A)\partial_{z^A}-\frac{\onman{\dot{a}}}{\onman{a}}z^A\partial_{z^A}.
\]
Accordingly, the vector $W_{\oo,\newX}$ associated to the frame $\{\newX_A(t)\}$
along the curve $\oo(t)$ reads as
\[
  W_{\oo,\newX} =\flow-\flow(z'^A)\partial_{z'^A}-\frac{\onman{\dot{a}}}{\onman{a}}z'^A\partial_{z'^A}.
\]
Therefore, using \eqref{eq:trans_z}, \eqref{eq:R} and \eqref{eq:trans_partial_z} we have
\begin{align*}
  W_{\oo,\newX} =
  &\flow-\flow(\onman{R}^A{}_C z^C)\onman{R}_A{}^B\partial_{z^B}
    -\frac{\onman{\dot{a}}}{\onman{a}}\onman{R}^A{}_B z^B\onman{R}_A{}^C\partial_{z^C}\\
  =&\flow-\delta^{AE}\delta_{CD}\flow(\onman{R}_E{}^D)z^C\onman{R}_A{}^B\partial_{z^B}
     -\onman{R}^A{}_C \flow(z^C)\onman{R}_A{}^B\partial_{z^B}
     -\frac{\onman{\dot{a}}}{\onman{a}}z^C\partial_{z^C}\\
  =&\flow+\delta^{AE}\delta_{CD}\onman{R}_E{}^F\onman{\FF}_F{}^D z^C\onman{R}_A{}^B\partial_{z^B}
     -\flow(z^B)\partial_{z^B}
     -\frac{\onman{\dot{a}}}{\onman{a}}z^C\partial_{z^C}\\
  =&\flow+\delta_{CD}\delta^{FB}\onman{\FF}_F{}^D z^C\partial_{z^B}-\flow(z^B)\partial_{z^B}
     -\frac{\onman{\dot{a}}}{\onman{a}}z^C\partial_{z^C}\\
  =&\flow+\onman{\FF}^B{}_C z^C\partial_{z^B}-\flow(z^B)\partial_{z^B}
     -\frac{\onman{\dot{a}}}{\onman{a}}z^C\partial_{z^C},
\end{align*}
and thus
\[
  W_{\oo,\newX} =\W+\onman{\FF}^B{}_C z^C\partial_{z^B}.
\]

\begin{remark}\label{rem:curve_F}
    If we write the RW metric in standard coordinates in the form
    \begin{align}
      \label{RWmet}
    g=- \signature d\tau^2+a^2(\tau)\left(1+\frac{\curvN}{4}|x|^2\right)^{-2}
      \delta_{AB} dx^A dx^B.
    \end{align}
    then the curve $\mathfrak{o}(t)$ can be described by a set of functions
    $\oo(t):=\{x^A=\mathfrak{o}^A(t)\}$. If we consider the simplest coordinate transformation that sends this stereographic  coordinates $\{x^A\}$ to
    stereographic coordinates $\{ z^A\}$ centered at $\mathfrak{o}(t)$, one can show that the metric \eqref{RWmet} gets transformed into the metric
    of Theorem \ref{main2} with the following expressions for
    the coefficients
  \begin{align*}
          &\frac{F_0{}^A}{a}=-\frac{1}{1+\frac{\curvN}{4}\oo^2}\,\dot \oo^A,\\
  &F^{BA}=\frac{\curvN}{1+\frac{\curvN}{4}\oo^2}\frac{1}{2}\left( \oo^B \dot\oo^A-\oo^A\dot \oo^B\right).
  \end{align*}
\end{remark}

\subsection{Explicit form of the Killings in the new coordinates}
\label{Sub:Killings}
For completeness, we also provide the expressions of the $\dim(\dim+1)/2$ Killing vector
fields in the chart $\{\tau,z^A\}$. There are several ways to approach the problem. A very direct one is to consider the following 
vector field anzatz
\begin{equation*}
  \xi= 
  \left\{\onman{\kill}_{0A}
    \frac{\curvN}{2}\left(z^Az^B-\frac{1}{2}\zsq \delta^{AB} \right)+\onman{\kill}_{0A}\delta^{AB}+\onman{\kill}_{AC}z^A\delta^{CB}\right\}
  \partial_{z^B}, 
\end{equation*}
where $\kill_{0A}(t)$,  $\kill_{AB}(t)=-\kill_{BA}(t)$ are  $\dim(\dim+1)/2$ free functions of $t$. 
A direct computation shows that $\xi$ is a Killing vector field of $g$
if and only if the system of differential equations 
\begin{align*}
  &\dot\kill_{0A}=
    \frac{1}{a}\FF_0{}^{C}\kill_{AC}
    -\FF_A{}^{C}\kill_{0C},\\ 
  &\dot\kill_{AB}=\frac{\curvN}{a}\left(\FF_{0A}\kill_{0B}-\FF_{0B}\kill_{0A}\right)
    +\FF_A{}^{C}\kill_{BC}
    -\FF_B{}^{C}\kill_{AC}
\end{align*}
is satisfied. The general solution depends on $\dim(\dim+1)/2$ free parameters, which matches the number of linearly independent Killing vectors of the metric. 
A basis of the  $\dim(\dim+1)/2$ dimensional algebra can be computed by choosing, for instance,
initial conditions given by the $\dim$ relations
\begin{equation}
  \{\kill_{0A}(t_0)=\delta_{AC},\quad  \kill_{AB}(t_0)=0\}\quad \mbox{ for } C=1,\ldots,\dim
\end{equation}
plus the $\dim(\dim-1)/2$ conditions
\begin{equation}
  \{\kill_{0A}(t_0)=0,\quad  \kill_{AB}(t_0)=\varepsilon_{ABC_1...C_{n-2}}\}\quad \mbox{ for } C_1,...,C_{n-2}=1,\ldots,\dim,
\end{equation}
where we have used $\varepsilon_{ABC_1...C_{n-2}}$ for the $\dim$-dim
antisymmetric symbol. It is immediate to check that the $\dim(\dim+1)/2$ vector fields thus constructed are linearly independent, and hence a basis of Killing vectors.

\section*{Acknowledgements}
We thank Jose J. Blanco-Pillado and Jon Urrestilla for useful comments
and a couple of references.
M.M. acknowledges financial support under
Grant PID2021-122938NB-I00 funded by MCIN/AEI /10.13039/501100011033 and
by “ERDF A way of making Europe” and RED2022-134301-T funded by
MCIN/AEI/10.13039/501100011033.
R.V. was supported by Grants IT956-16 and IT1628-22 from the Basque Government,
and PID2021-123226NB-I00 funded
by  “ERDF A way of making Europe” and MCIN/AEI/10.13039/501100011033.
Some algebraic computations have been performed with the help of the free PSL version of REDUCE.
\appendix

\section{Constant curvature spaces: charts and conformal Kil\-ling fields}
\label{sec:constant_curv_spaces}

In this appendix we find a canonical form for the metric
in a semi-Riemannian space of constant
curvature $(S,\g)$ of dimension $\dim$. Given 
any point $o \in S$ and
a basis  $\{X_A\}$ of $T_o S$ we construct canonically a set of $\dim$ functions $x^A$ on a neighbourhood $U_o$ of $o$ such that the metric
is given by \eqref{confflat}, where $\curvS$ is the curvature of $(S,\g)$ and the constants $h_{AB}$ are defined
by $h_{AB} := \g|_o (X_A, X_B)$. The result is obviously known. However, we include its derivation for two reasons.
First, the argument is self-contained, elegant and short (and, to the best of our knowledge, also new). Second, because the construction being
explicit and based on solutions of certain system of overdetermined system of PDE, it allows us to show in the main text how to construct a coordinate system in the foliated manifold.

The construction is based on the following definition.
\begin{definition}\label{def:hess_basis}
  A semi-Riemannian manifold $(S,\g)$ of dimension $\dim \geq 1$
  admits a {\bf Hessian basis} if there exist $\dim+2$ real functions $\{\YyS^{\ta}\}
  = \{ \YyS^0,\YyS^A,\YyS^{\dim+1}\}$ on $S$ that solve 
\begin{align}
&  \Hess \YyS^{\ta}= \left( -\curvS \YyS^{\ta} + \delta^{\ta}_{\dim+1}\right) \gS \label{Hesseq}\\
& \YyS^{\ta} |_{o} = \delta^{\ta}_0, \quad \qquad
X_A(\YyS^{\ta}) |_{o} = \delta^{\ta}_A. \label{boundary}
\end{align}
where $\Hess$ is the Hessian with respect to $\gS$, $\curvS \in \mathbb{R}$,
$o$ is  any point $o \in S$ and $\{X_A\}$ is any basis of $T_o S$.
Observe that given $\{X_A\}$, the Hessian basis is unique.
\end{definition}

In this definition $(S,\g)$ is not assumed to be of constant
curvature $\curvS$. However, in Remark \ref{constant} below we show that
whenever $\dim \geq 2$ this is a necessary and (locally) sufficient requirement for the existence of a Hessian basis.

The following result constructs the coordinates $\{x^A\}$ mentioned above.
\begin{lemma}
\label{properties}
  Let $(S,\g)$ be connected and admit a Hessian basis $\{ \YyS^{\ta}\}$. Let  $\g^{AB}$ stand for the contravariant metric of $\g$ and define 
the following two 
  $(\dim+2) \times (\dim+2)$  matrices of functions
 ($\vec{0}$ is a column with $m$ zeroes
and $\vec{0}^{\,T}$ its transpose, $\ta$ is a row index and
$\tb$ a column index)
  \begin{align*}
C^{\ta \tb} :=    \left ( \begin{array}{ccc}
              \curvS & \vec{0}^{\,T} & -1 \\
                            \vec{0} & \gS^{AB} & \vec{0} \\
                            -1 & \vec{0}^{\,T} & 0
                          \end{array}
                                          \right ), 
\qquad     A^{\ta}{}_{\tb}  := 
\left ( \begin{array}{ccc}
                                      \YyS^{\ta} & D_A \YyS^{\ta} &
\delta^{\ta}_{\dim+1} 
\end{array} \right ).
  \end{align*}
Then,
\begin{itemize}
\item[(i)] The functions  $Q^{\ta\tb} := A^{\ta}{}_{\tc} A^{\tb}{}_{\td} C^{\tc\td}$ 
are constant on $S$ and, in fact, 
  \begin{align}
Q^{\ta \tb} :=    \left ( \begin{array}{ccc}
              \curvS & \vec{0}^{\,T} & -1 \\
                            \vec{0} & h^{AB} & \vec{0} \\
                            -1 & \vec{0}^{\, T} & 0
                          \end{array}
                                          \right ),
\label{Qexp}
\end{align} 
  where $h^{AB}$ is the inverse of $h_{AB}$. Moreover, the
  matrix $A^{\alpha}{}_{\beta}$ is invertible at every point.
\item[(ii)] The following algebraic relations hold:
\begin{align}
\YyS^0 + \curvS \YyS^{\dim+1} &= 1, \label{algeb1}\\
\ysq
- 2 \YyS^{\dim+1} + \curvS \left ( \YyS^{\dim+1} \right )^2 & =0, \label{algeb2}
\end{align}
where $\ysq := h_{AB} \YyS^A \YyS^B$.
\item[(iii)] The metric $\g$ can be written everywhere as 
\begin{align}
  \g = h_{AB} d \YyS^A \oti d \YyS^B + \curvS d \YyS^{\dim+1} \oti d \YyS^{\dim+1}.
  \label{induced}
\end{align}
\item[(iv)] The functions
  \begin{align}
    x^A := \frac{\YyS^A}{1- \frac{1}{2} \curvS \YyS^{\dim+1}} \label{defxa}
  \end{align}
  are well-defined on $S':=S \setminus S_0$, where $S_0 := \{ \ysq  =0\} \cap \{\YyS^{\dim+1} \neq 0\}$. At every point in $S'$
  it holds
\begin{align}
1 + \frac{\curvS}{4} \xsq  \neq 0, \qquad
\gS  = \frac{1}{\left ( 1 + \frac{\curvS}{4} \xsq \right )^2} h_{AB} dx^A \oti dx^B,  \label{confflat}
\end{align}
where 
$\xsq := h_{AB} x^A x^B$.
In particular $\{ x^A\}$ defines a local coordinate system on $S'.$
\end{itemize}
\end{lemma}

\begin{remark}\label{embedding}
  The usual description of constant curvature spaces
  as embeddings in  $\mathbb{R}^{n+1}$ are recovered as follows.
  Observe that \eqref{algeb1} and \eqref{algeb2} imply
  \begin{align}
    \label{graph}
    \curvS\ysq + (\YyS^0)^2=1.
  \end{align}
  Moreover $d \YyS^{0} = -\kappa d\YyS^{n+1}$. Thus, when $\kappa\neq0$ the metric \eqref{induced} can be viewed as the induced metric of the surface defined by
  \eqref{graph} embedded in  $\mathbb{R}^{n+1}$ with the metric
  $h_{AB} d\YyS^A d\YyS^B + \kappa^{-1} (d\YyS^0)^2$ (when $\kappa=0$ we can simply take the Euclidean metric $h_{AB} dy^A dy^B + (d\YyS^0)^2$ since the graph is now simply $\YyS^0=\pm 1$).
  Note that coordinating the surface with $\{ \YyS^A\}$ (so that it is described as a graph in the coordinate $\YyS^0$) does not cover the whole surface, but just a hemisphere, when $\kappa > 0$. Note also that viewing $\gamma$ as an induced metric  in a one-dimensional higher space requires a different treatment depending on the value of $\kappa$.
\end{remark}

\begin{remark}\label{res:stereo_cords}
  The functions $\{ x^A \}$ of this lemma are called {\it stereographic coordinates centered at $o$ with frame $\{X_A\}$}. Given the frame $\{X_A\}$, the functions $\{ x^A \}$ are unique.
\end{remark}

\begin{proof}
The definition of $Q^{\ta\tb}$ entails
\begin{align}
Q^{\ta\tb} =  \g \left (\gradS \YyS^{\ta}, 
\gradS \YyS^{\tb} \right ) + \curvS \YyS^{\ta} \YyS^{\tb}
- \YyS^{\ta}\delta^{\tb}_{\dim+1}
- \YyS^{\tb}\delta^{\ta}_{\dim+1}, \label{Qalt}
\end{align}
and its gradient is (in abstract index notation and letting
$D_A$ be the Levi-Civita covariant derivative to $\gS$) 
\begin{align*}
D_A Q^{\ta\tb} = D_A D_B \YyS^{\ta} D^B \YyS^{\tb} + 
D_A D_B \YyS^{\tb} D^B \YyS^{\ta} 
+ D_A \YyS^{\ta} \left ( 
\curvS \YyS^{\tb} - \delta^{\tb}_{\dim+1} \right )
+ D_A \YyS^{\tb} \left ( 
\curvS \YyS^{\ta} - \delta^{\ta}_{\dim+1} \right ) =0
\end{align*}
after inserting equation  \eqref{Hesseq}.
So $Q^{\ta\tb}$ is constant on $S$.
At $o$ we have $\YyS^{\ta} = \delta^{\ta}_0$ and
\begin{align*}
\g|_o \left (\gradS \YyS^{\ta}, \gradS \YyS^{\tb} \right ) 
= h^{AB} X_A (\YyS^{\ta}) X_B (\YyS^{\tb} ) = h^{AB}\delta^{\ta}_A \delta^{\tb}_B 
\end{align*}
because of \eqref{boundary}. Evaluating \eqref{Qalt} at $o$ gives \eqref{Qexp}.

Now, the matrix $C^{\ta\tb}$ is invertible everywhere (its determinant is up to a sign that of the contravariant metric $\g^{\sharp}$). $Q^{\ta\tb}$ is also invertible because of \eqref{Qexp}. This means that $A^{\ta}{}_{\tb}$ is invertible at every point. Writing
the inverses
of $C^{\ta\tb}$, $Q^{\ta\tb}$ as $C_{\ta\tb}$, $Q_{\ta\tb}$ one gets
\begin{align}
Q_{\ta\tb} A^{\ta}{}_{\tc} A^{\tb}{}_{\td} = C_{\tc \td}
\label{QAAC}
\end{align}
as an immediate consequence  of the definition of $Q^{\ta\tb}$ in item (i). The inverses of $C^{\ta\tb}$ and
$Q^{\ta\tb}$ are
\begin{align*}
C_{\ta \tb} :=    \left ( \begin{array}{ccc}
              0 & \vec{0} & -1 \\
                            \vec{0}^{\,T} & \gS_{AB} & \vec{0}^{\,T} \\
                            -1 & \vec{0} & - \curvS
                          \end{array}
                                          \right ), 
\qquad 
Q^{\ta \tb} :=    \left ( \begin{array}{ccc}
              0 & \vec{0} & -1 \\
                            \vec{0}^{\,T} & h_{AB} & \vec{0}^{\,T} \\
                            -1 & \vec{0} & -\curvS
                          \end{array}
                                          \right ).
\end{align*}
Evaluating \eqref{QAAC} at the values $\tc =0, \td = \dim+1$ gives
\eqref{algeb1} and at the values $\tc = \td =0$ gives 
\begin{align*}
h_{AB} \YyS^A \YyS^B - 2 \YyS^0 \YyS^{\dim+1} - \curvS (\YyS^{\dim+1})^2 =0,
\end{align*}
which becomes \eqref{algeb2} after inserting \eqref{algeb1}.
Evaluating \eqref{QAAC}  at $\tc = C$ and $\td = D$ provides
\begin{align*}
& \g_{CD} = h_{AB} D_C \YyS^A D_D \YyS^B 
- 2 D_C \YyS^0 D_D  \YyS^{\dim+1}
- 2 D_C \YyS^{\dim+1} D_D  \YyS^{0}
- \curvS D_C \YyS^{\dim+1}  D_D \YyS^{\dim+1}
\quad \\
& \Longleftrightarrow
\quad
\g = h_{AC} d \YyS^A \oti d \YyS^{B}
- d \YyS^0 \oti d \YyS^{\dim+1}
- d \YyS^0 \oti d \YyS^{\dim+1}
- \curvS d \YyS^{\dim+1} \oti d \YyS^{\dim+1}
\end{align*}
which proves item (iii) after using
$d\YyS^0 = - \curvS d\YyS^{\dim+1}$ (cf. \eqref{algeb1}).

For the last item we first note that $x^A$ is well-defined
at every point with $\YyS^{\dim+1} =0$ (in particular at $o$, 
where $x^A|_o =0$). At points where $\YyS^{\dim+1} \neq 0$ we can write $x^A$ alternatively as
\begin{align}
x^A =\frac{\YyS^A}{1- \frac{1}{2} \curvS \YyS^{\dim+1}}
=  \frac{\YyS^{\dim+1} \YyS^A}{\YyS^{\dim+1}- \frac{1}{2} \curvS (\YyS^{\dim+1})^2}
=  \frac{2 \YyS^ {\dim+1} \YyS^{A}}{\ysq}, \label{xalt}
\end{align}
where in the last equality we inserted \eqref{algeb2}. This expression shows that $x^A$ is well-defined everywhere except on the set $S_0$  defined in item (iv). We next prove that
\begin{align}
  1 + \frac{\curvS}{4} \xsq 
  = \frac{1}{1 - \frac{\curvS}{2} \YyS^{\dim+1}} \qquad
  \mbox{on } \quad S \setminus S_0.
\label{xsq}
\end{align}
At any point where $\YyS^{\dim+1}=0$ this is immediate because $x^A = \YyS^A$ and 
$\ysq =0$ (cf. \eqref{algeb2}). At any point
where $\YyS^{\dim+1} \neq 0$ we use \eqref{xalt}
\begin{align*}
1 + \frac{\curvS}{4} \xsq =  1 + \frac{\curvS (\YyS^{\dim+1})^2}
{\ysq}
= \frac{1}{\ysq} \left ( 
\ysq + \curvS (\YyS^{\dim+1})^2 \right ) \stackrel{\eqref{algeb2}}{=} \frac{2 \YyS^{\dim+1}}{2 \YyS^{\dim+1} - \curvS (\YyS^{\dim+1} )^2},
\end{align*}
which is \eqref{xsq}. It only remains to prove
that the right-hand-side is finite at any $p \in S \setminus S_0$ (the fact that it is non-zero is obvious). By definition of $S_0$ 
either $\ysq|_p  \neq 0$ or $\YyS^{\dim+1}|_p =0$. In either case $1 - \frac{\curvS}{2} \YyS^{\dim+1} \neq 0$ (recall \eqref{algeb2})  and finiteness follows.
Thus, \eqref{confflat} is equivalent to 
\begin{align}
  h_{AB} dx^A \oti dx^A = \frac{1}{\left ( 1- \frac{\curvS}{2} \YyS^{\dim+1} \right )^2} \gS. \label{confflat2}
\end{align}
Computing $dx^A$ from \eqref{defxa}
and using $h_{AB} \YyS^A d\YyS^B = \left ( 1 - \curvS \YyS^{\dim+1} \right )d\YyS^{\dim+1}$ (i.e. the differential of \eqref{algeb2}), \eqref{confflat2} follows by direct computation.

The last statement of the lemma is immediate from the form \eqref{confflat}. Explicitly, let $p$ be any point on $S \setminus S_0$ and $\{ z^A\}$ a coordinate system near $p$. We work on the domain of this chart. The Jacobian matrix
$J^A{}_{B} := ( \frac{\partial x^A}{\partial z^B} )$ must have non-zero
determinant  because ($\gamma^{(z)}$ stands for the components of the metric in the chart $\{z^A\}$)
\begin{align*}
\g^{(z)}_{CD} = \frac{1}{\left ( 1+ \frac{\curvS}{4} \xsq \right )^2}
h_{AB} J^{A}{}_C J^{B}{}_D.
\end{align*}
\end{proof}
\begin{remark}
  \label{constant}
  When $\dim \geq 2$ the existence of a Hessian basis requires that $(S,\gS)$ is of constant curvature $\curvS$. Indeed, by the Ricci identity
  \begin{align*}
    -R^E{}_{ABC}D_E \YyS^{\alpha} =
    D_B D_C D_A \YyS^{\alpha}
    - D_C D_B D_A \YyS^{\alpha}
    = \curvS \left ( \gS_{CA} \delta^E{}_B- \gS_{BA} \delta^E{}_C \right ) D_E \YyS^{\alpha},
  \end{align*}
 and the result follows because the $(\dim+2)\times \dim$ matrix  $\left ( D_F \YyS^{\alpha} \right )$ has rank $\dim$ (if it had lower rank, then
  $A^{\alpha}{}_{\beta}$ would not be invertible, against item (i)
  in Lemma \ref{properties}).

  Conversely, if $(S,\gS)$ has $\dim \geq 2$
  and constant curvature $\curvS$,
  then there exists a sufficiently small neighbourhood around any point $p$  
  (with the induced metric) that admits a Hessian basis, as we show next.
  \begin{proof}[Proof of existence of a Hessian basis]
      The integrability conditions of the equation $\Hess f=F\gS$
      with $F=-\curvS f+q_0$ for $q_0\in\mathbb{R}$
      are identically satisfied because
      \begin{align*}
        &\gS_{CA} D_BF-\gS_{BA}D_CF-2D_{[B}D_{C]}D_A f=\gS_{CA} D_BF-\gS_{BA}D_CF-R^E{}_{ABC}D_Ef\\
        &=\gS_{CA} D_BF-\gS_{BA}D_CF+\curvS \left ( \gS_{CA} \delta^E{}_B - \gS_{BA} \delta^E{}_C \right )D_Ef=0.
      \end{align*}
      Therefore, fixing any $p\in S$ there is one solution for each choice of $f|_p$ and $df|_p$,
      that is $\dim+1$ constants. This yields $\dim +1$ linearly independent solutions.
      These solutions exist in a suitable small neighbourhood of $p$.
    \end{proof}
  \end{remark}

For use in the main text, we next show that on $S\setminus S_0$ the Hessian basis $\{ \YyS^{\alpha}\}$ can be written in
terms of $x^A$ as follows
\begin{align}
  \YyS^A =
  \frac{x^A}{1 +\frac{\curvS}{4} \xsq},
  \qquad
  \YyS^{\dim+1} = \frac{\xsq}{2
    \left (1 +\frac{\curvS}{4} \xsq\right )}, \qquad
  \YyS^0
  = \frac{1 -\frac{\curvS}{4} \xsq}{1 +\frac{\curvS}{4} \xsq}. \label{Yinx}
\end{align}
The first equality follows directly from the \eqref{defxa} and \eqref{xsq}.  The second is obtained by squaring \eqref{defxa}
\begin{align*}
  \xsq = \frac{\ysq}{
    \left ( 1 - \frac{\curvS}{2} \YyS^{\dim+1}
    \right )^2}
  = \frac{2 \YyS^{\dim+1}}{
    \left ( 1 - \frac{\curvS}{2} \YyS^{\dim+1}
    \right )} = 2 \YyS^{\dim+1}
 \left (1 +\frac{\curvS}{4} \xsq\right ),
\end{align*}
where in the second equality we used
\eqref{algeb2} and in the third \eqref{xsq}.
The equality for $\YyS^{0}$ follows from  \eqref{algeb1}.

The last result of the Appendix gives properties of the conformal Killing algebra of $(S,\gS)$ that are needed in the main text. Although the explicit
expressions \eqref{ckvl_S}-\eqref{killsl_S} are of course known, we include the derivation in order to make the paper self-contained and because
its construction is almost immediate when one uses Hessian basis (see
\cite{MarsPeonNieto} for a related construction).
\begin{lemma}\label{res:CKVs}
  Let $(S,\gS)$ be of dimension $\dim \geq 2$ and
  admit  a Hessian basis $\{ \YyS^{\alpha}\}$. Then, the vector
  fields
  \begin{equation}
  \ckvS^{\ta\tb}\defi\YyS^{\ta} \gradS \YyS^{\tb}-\YyS^{\tb} \gradS \YyS^{\ta},
  \label{def:zeta_S}
  \end{equation}
    are conformal Killing vectors of $\gS$ satisfying
  \begin{align}
  \pounds_{\ckvS^{\ta\tb}}\gS=
    2(\YyS^{\ta}\delta^{\tb}_{\dim+1}-\YyS^{\tb}\delta^{\ta}_{\dim+1})\gS.
  \label{eq:Lie_zeta}
  \end{align}
   Moreover, if $\dim \geq 3$ then the set $\{ \ckvS^{\ta\tb}, \ta <\tb\}$ is a basis of the conformal Killing algebra of $(S,\gS)$. On $S \setminus S_0$ they read,  in the local coordinates
  $\{ x^A\}$,
\begin{align}
  &\ckvS^{0\, \dim+1}=\lS^A\partial_A,\qquad
    && \ckvS^{A\, \dim+1}
    =\left(\lS^A\lS^B-\frac{1}{2}\lS^2 h^{AB} \right)\partial_B,\label{ckvl_S}\\
    &\ckvS^{AB}= \left ( \lS^A h^{BC}- \lS^B h^{AC} \right )  \partial_C,\quad
    && \ckvS^{0A}  = h^{AB} \partial_B+\frac{\curvS}{2} \ckvS^{A\, \dim+1}.\label{killsl_S}
\end{align}
  \end{lemma}

\begin{proof}
  Using square brackets for antisymmetrization and brackets for symmetrization, equation \eqref{Hesseq} implies
  \begin{align*}
2 D_{(A} \ckvS_{B)}^{\ta\tb} =   4  D_{(A} (\YyS^{[\ta} D_{B)}\YyS^{\tb]}) = 4 D_{(A}\YyS^{[\ta}D_{B)}\YyS^{\tb]}
+ 4 \YyS^{[\ta}\delta^{\tb]}_{\dim+1} \gS_{AB}
= 4 \YyS^{[\ta}\delta^{\tb]}_{\dim+1} \gS_{AB},
  \end{align*}
which is just another way of writing \eqref{eq:Lie_zeta}. The gradients of $\{\YyS^{\alpha}\}$ on $S \setminus S_0$
in the local coordinates $\{ x^A\}$ 
are computed immediately from \eqref{Yinx} after noting that
\begin{align*}
  \gradS x^A = \left (1 + \frac{\curvS}{4} \xsq \right )^2 h^{AB}
  \partial_B, \qquad
  \gradS \xsq = 2
  \left (1 + \frac{\curvS}{4} \xsq \right )^2 x^B \partial_{B}.
\end{align*}
The result is
\begin{align}
  &\gradS \YyS^0= - \curvS\lS^B\partial_B,
    \quad
    \gradS \YyS^A=\left(\left(1+\frac{\curvS}{4}\lS^2\right)h^{AB}
    -\frac{\curvS}{2}
    \lS^A  \lS^B\right)\partial_B, \quad 
    \gradS \YyS^{\dim+1}=\lS^B\partial_B\label{grads_Yil_S}.
\end{align}
Inserting \eqref{Yinx} and \eqref{grads_Yil_S} into \eqref{def:zeta_S} yields
\eqref{ckvl_S}-\eqref{killsl_S} after a simple computation.
It only remains to show that $\{\ckvS^{\ta\tb}; \ta < \tb\}$ is a basis. Since this set has $(\dim+1)(\dim+2)/2$  elements and this is the maximal  dimension of a conformal Killing algebra in dimension $\dim \geq 3$, it suffices to prove that they are linearly independent vector fields. Actually, it  suffices to prove this on some non-empty open set.
It is immediate to check that \eqref{ckvl_S}-\eqref{killsl_S} are linearly independent on $S \setminus S_0$, so the result follows.
\end{proof}

\bibliography{references_v7}{}
\bibliographystyle{myieeetr}

\end{document}